\documentclass[aps,amsfonts,twoside,amssymb,superscriptaddress,twocolumn]{revtex4-1}
\usepackage{amsfonts,amssymb,amsmath,amsthm}
\usepackage[latin1]{inputenc}
\usepackage[T1]{fontenc}
\usepackage{amsmath,amssymb,amstext,amsfonts,amsxtra}
\usepackage{amsthm}
\usepackage{graphicx}
\usepackage{mathrsfs}
\usepackage{enumerate}
\usepackage{bbm}
\usepackage{times}
\usepackage{hyperref}
\usepackage{color}
\usepackage{braket}

\newcommand{\tr}{\operatorname{tr}}
\def\idty{{\leavevmode\rm 1\mkern -5.4mu I}} 
\def\id{{\rm id}}

\def\ket #1{\vert #1\rangle}

\def\ketbra #1#2{\vert #1\rangle \langle #2\vert}

\def\tr{\mathop{\rm tr}\nolimits}

\newcommand*{\cC}{\mathcal{C}}

\newcommand*{\cE}{\mathcal{E}}
\newcommand*{\cF}{\mathcal{F}}

\newcommand*{\cH}{\mathcal{H}}

\newcommand*{\cM}{\mathcal{M}}

\newcommand*{\cN}{\mathcal{N}}

\newcommand*{\cO}{\mathcal{O}}
\newcommand*{\cP}{\mathcal{P}}

\newcommand*{\cX}{\mathcal{X}}

\newcommand*{\cZ}{\mathcal{Z}}

\newcommand*{\dd}{\textrm{d} }

\newcommand*{\cl}{\textrm{cl}} 
\newcommand*{\inn}{\textrm{in}} 
\newcommand*{\out}{\textrm{out}} 
\newcommand*{\suc}{\textrm{succ}} 
\newcommand*{\ec}{\textrm{EC}} 
\newcommand*{\av}{\textrm{av}} 
\newcommand*{\maj}{\textrm{Maj}} 
\newcommand*{\gauss}{\textrm{Gauss}} 
\newcommand*{\iid}{\textrm{IID}} 
\newcommand*{\ther}{\textrm{th}} 
\newcommand*{\ot}{\textrm{OT}} 
\newcommand*{\asym}{\textrm{as}}

\def\PP{{\mathbb P}}

\newcommand\pr[1]{\ensuremath{\mathrm{Pr}[#1]}}

\newtheorem{thm}{Theorem}[section]

\newtheorem{de}[thm]{Definition}
\newtheorem{lem}[thm]{Lemma}

\def\cH{{\mathcal H}}

\usepackage{color}
\definecolor{myred}{rgb}{1,0,0}

\definecolor{myblue}{rgb}{0,0,0.8}

\definecolor{myyellow}{rgb}{0.9,0.8,0}

\definecolor{mygreen}{rgb}{0,0.6,0}

\definecolor{myorange}{rgb}{0.6,0.6,0}

\definecolor{mycerul}{rgb}{0,0.6,1}

\begin{document}

\title{Continuous-Variable Protocols in the Noisy-Storage Model}

\author{Fabian Furrer}
\affiliation{NTT Basic Research Laboratories, NTT Corporation, 3-1 Morinosato-Wakamiya, Atsugi, Kanagawa, 243-0198, Japan. } 
\affiliation{Department of Physics, Graduate School of Science,
University of Tokyo, 7-3-1 Hongo, Bunkyo-ku, Tokyo, Japan, 113-0033.}

\author{Christian Schaffner} 
\affiliation{Institute for Logic, Language and Computation (ILLC)
University of Amsterdam, The Netherlands}
\affiliation{Centrum Wiskunde \& Informatica (CWI), Amsterdam, The Netherlands}

\author{Stephanie Wehner} 
\affiliation{QuTech, Delft University of Technology, Lorentzweg 1, 2628 CJ Delft, Netherlands}

\begin{abstract}
We present the first protocol for oblivious transfer that can be implemented with an optical continuous-variable system, and prove its security in the noisy-storage model. This model allows security to be achieved by sending more quantum signals than an attacker 
can reliably store at one specific point during the protocol. 
Concretely, we determine how many signals need to be sent in order to achieve security by 
establishing a trade-off between quantum uncertainty generated in the protocol and the classical capacity of the memory channel. As our main technical tool, we study and derive new uncertainty relations for continuous-variable systems. Finally, we provide explicit security parameters for realistic memory models.
\end{abstract}

\maketitle
\section{Introduction}
Quantum key distribution (QKD) offers security that rests only on the laws of quantum mechanics~\cite{Wiesner83,Bennett84,Ekert91}. 
Yet, there are still important cryptographic protocols which cannot be realized without additional assumptions, even using quantum communication~\cite{mayers1997,mayers:trouble,lo&chau:bitcom,lo&chau:bitcom2,lo1997,kretch:bc,buhrman2012complete}. Examples of such protocols are oblivious-transfer (OT), bit commitment (BC), and secure password-based identification, where two distrustful parties (Alice and Bob) engage in a protocol and want to be ensured that the other party cannot cheat, or maliciously influence the outcome.

Due to the great practical importance of problems such as secure identification one is willing to rely on assumptions in order to achieve security. Classically, these are usually computational assumptions that are not fully future proof and can possibly be broken retroactively using a quantum computer. Another line of research pursues physical assumptions on the adversary, such as imposing limits on his abilities to
store information~\cite{Maurer92b,cachin:bounded}. This assumption is especially appealing in quantum communication where advanced technologies such as long-lived quantum memories are very challenging and expensive. Moreover, as opposed to computational assumptions, it provides the benefit that they are indeed fully future proof: even if the adversary obtains a much larger quantum memory after the protocol, security cannot be broken retroactively. 

Given any constraint on the adversary's storage device, security can always be obtained by sending sufficiently many signals during the course of the protocol. 
Generalizing the model of classical 
bounded storage~\cite{Maurer92b}, the so-called bounded-quantum-storage model assumes that the adversary can only store a certain number of 
qubits~\cite{damgaard2008,damgaard2007}. More generally, the noisy-storage model~\cite{wehner2008} ensures security for arbitrary noisy memory devices. Specifically, a link has been made between security and the classical capacity~\cite{Koenig2012}, entanglement cost~\cite{Berta2012}, 
and quantum capacity~\cite{berta2013, Dupuis2015} of the adversary's quantum storage device.
An important feature of the corresponding protocols is that they do not require the honest protocol participants to have any quantum memory at all.
In particular, they can be implemented using BB84 or six-state QKD protocols, which have been experimentally demonstrated~\cite{ng2012,erven2014experimental}.  

Yet, all protocols proposed so far~\cite{damgaard2008,wehner2008,Koenig2012,wcsl2010,schaffner2010,ng2012,Berta2012} are based on discrete-variable systems requiring single-photon detectors that are, despite recent improvements, still challenging technologies~\cite{lo2014}. Here, we propose the first protocols based on optical continuous-variable systems, where the information is encoded in the $X$ and $P$ quadrature of the electromagnetic field (see, e.g.,~\cite{Weedbrook12}). Optical continuous variable implementations provide practical benefits  since transmission, measurements (homodyne detection) and sometimes also preparations require only standard telecommunication technologies. These benefits allow easy integration of the protocols into current classical networks. Moreover, state preparation and homodyne detection are highly efficient and robust technologies permitting high clock rates, and they are available as on-chip components~\cite{masada2015}.

We present a protocol for OT as well as BC, and derive sufficient conditions for security in the noisy-storage model depending on the classical capacity of the malicious party's memory channel similar to~\cite{Koenig2012}. For instance, we show that security can be obtained if twice the classical capacity is lower than the uncertainty that is generated by $X$ and $P$ measurements plus the error-correction (EC) amount that is required to overcome the information loss during transmission. 
The latter term is crucial for CV protocols since, compared to discrete variable protocols, CV protocols require the exchange of a significant amount of EC information.  

The main technical ingredients in our security proof are novel entropic uncertainty relations. While we derive an uncertainty relation that holds without assumptions by using majorization techniques from~\cite{rudnicki2015}, it turns out that even though security is in principle possible it is not sufficient to obtain a good trade-off in parameters. We overcome this technical problem by showing uncertainty relations under reasonable assumptions, namely, that the adversary's encoding is Gaussian or independent and identical over only a limited number of modes. The security trade-off is then analyzed in both cases for a memory channel modeled by a lossy and noisy bosonic Gaussian channel. Our work opens the door for the development of continuous-variable protocols in the noisy-storage model.

\section{ Oblivious transfer in the noisy-storage model}
We first focus on OT along which we discuss the essential ideas behind the security in the noisy-storage model. A protocol for OT is especially appealing since any two-party cryptographic problem in which Alice and Bob do not trust each other can in principle be solved using OT as a building block~\cite{kilian1988founding}.
While the quantum part of the protocols for OT and BC protocol are similar, the classical post-processing is different. We consider a randomized version of OT, where Alice has no input and gets as output two bit strings $s_0,s_1$ and Bob has input $t$ and obtains a bit string $\tilde s$. If both are honest, we require that $\tilde s=s_t$ (correctness). If Alice is honest, we require that Bob can only know one of the strings. If Bob is honest, we demand that Alice does not learn $t$. No requirements are made if both are dishonest.
These security requirements are made precise in a composable fashion in~\ref{app:OTSecDef} by demanding that it is indistinguishable from a perfectly correct protocol with probability $\epsilon_C$, and from a perfectly secure protocol for honest Alice and Bob with $\epsilon_A,\epsilon_B$.

The quantum part of the protocol for CV is based on a QKD protocol using Gaussian modulated squeezed states.
As it is conceptually simpler, we consider an entanglement based version of the protocol, although its prepare and measure version is straightforward. The source is given by the CV equivalent of a maximally entangled state, namely, a  two-mode squeezed state simply referred to as EPR (Einstein,Podolski,Rosen~\cite{epr35}) state in the following (see Appendix~\ref{app:EPR}). The measurements are homodyne detections with a discretized outcome range into binnings of length $\delta$. Mathematically, they are modeled as coarse-grained $X$ and $P$ measurements and in the following denoted by $X_\delta$ and $P_\delta$ (see, e.g.,~\cite{furrer2014PQUR}). The measurement choice of Alice (Bob) in round $i$ are denoted by $\theta_A^i$ ($\theta_B^i$), where $\theta_A^i=0,1$ stands for performing $X_\delta,P_\delta$. The quantum protocol is then simply:
\begin{itemize}
\item[(Q1)] Alice creates $n$ EPR pairs of which she sends each half to Bob. 
\item[(Q2)] Alice and Bob measure independently $X_\delta$ ($\theta = 0$) or $P_\delta$ ($\theta=1$) according to $\theta_A = \theta_A^1\ldots\theta_A^n$ and $\theta_B = \theta_B^1\ldots\theta_B^n$, where $\theta_A$ and $\theta_B$ are chosen uniformly random in $\{0,1\}^n$~\footnote{Note that due to transmission losses in Bob's mode, he scales his outcomes with $1/\sqrt{\tau}$, where $\tau$ is the transmissivity.}. The strings of outcomes for Alice and Bob are denoted by $Z$ and $Y$.     
\item[(Q3)]They wait for a fixed time $\Delta t$. 
\end{itemize}
The crucial property of this protocol is that approximately half of the strings $Z$ and $Y$ are strongly correlated. However, neither party knows which part. This concept has been formalized in~\cite{Koenig2012} under the name of \emph{weak string erasure}. 

The classical part of the protocol proceeds as follows:
\begin{itemize}
\item[(OT1)] Alice sends Bob her basis choice $\theta_A$. Bob defines the set $I_t=\{i \mid \theta_A^i = \theta_B^i \}$ and its complement $I_{1-t}$ according to his choice bit $t$. He sends $I_0,I_1$ to Alice. 
\item[(OT2)] Alice forms the strings $Z_k= (Z^i)_{i\in I_k}$ for $k=0,1$, and computes error-correction information $W_0,W_1$ individually for $Z_0,Z_1$ and sends it to Bob. Bob then corrects the string corresponding to his choice $Y_t=(Y^i)_{i\in I_t}$ using $W_t$ to obtain $Y'_t$.~\footnote{Note that we could introduce an additional step that Bob can check if the error correction worked properly, namely, by Alice sending a hash of $Z_0,Z_1$. However, Bob is not allowed to tell Alice whether the test was passed or not, since Alice could design attacks which lead to pass or failure of the test depending on his choice $t$.} 
\item[(OT3)] Alice selects random $2$-universal hash functions $f_0,f_1$ from $X_0$ and $X_1$ to $\ell$-bit strings and outputs $s_k=f_k(X_k)$, $k=0,1$. She then sends $f_0,f_1$ to Bob who outputs $ \tilde s= f_t(Y'_t)$.  
\end{itemize}
We further assume that if a honest party obtains a value from the other party that is not conformal with the protocol, it generates a random output. This ensures that the protocol always terminates with an output.

It is easy to verify that executing the above classical protocol after (Q1)-(Q3) satisfies the correctness condition for OT with $\epsilon_C$ that depends on the error-correction protocol.  Moreover, security for honest Bob simply follows since the only information leaving his lab are the sets $I_0,I_1$, which are uncorrelated with $t$. For a more rigorous proof of  composable security, we refer to~\cite{damgaard2007}.

More interesting is the security for honest Alice. In fact, if Bob has a quantum memory that allows him to faithfully store all the modes sent by Alice over a time longer than $\Delta t$,  he can cheat perfectly. He only has to wait to receive Alice's basis choice $\theta_A$ and measure all modes in the corresponding basis. But if Bob has only noisy quantum storage he might not have enough information to obtain both strings $s_0$ and $s_1$. As illustrated in Fig.~\ref{fig:Memory}, we model Bob's memory attack by an encoding operation $\cE$ that maps the $n$ modes to the input space $Q_\inn$ of his quantum memory correlated to classical information $K$. After that Bob stores $Q_\inn $ in his quantum memory for time $\Delta t$ until he receives $\theta_A$. We model the corresponding memory channel $\cM$ of $\nu n$ quantum channels $\cF$, i.e., $\cM=\cF^{\otimes \nu n}$. Bob's information at the end of the protocol is denoted by $B'$ and given by all the classical information obtained from Alice plus $K$ and $\cM(Q_\inn)$. 

\begin{figure}\begin{center}\includegraphics*[width=8.5cm]{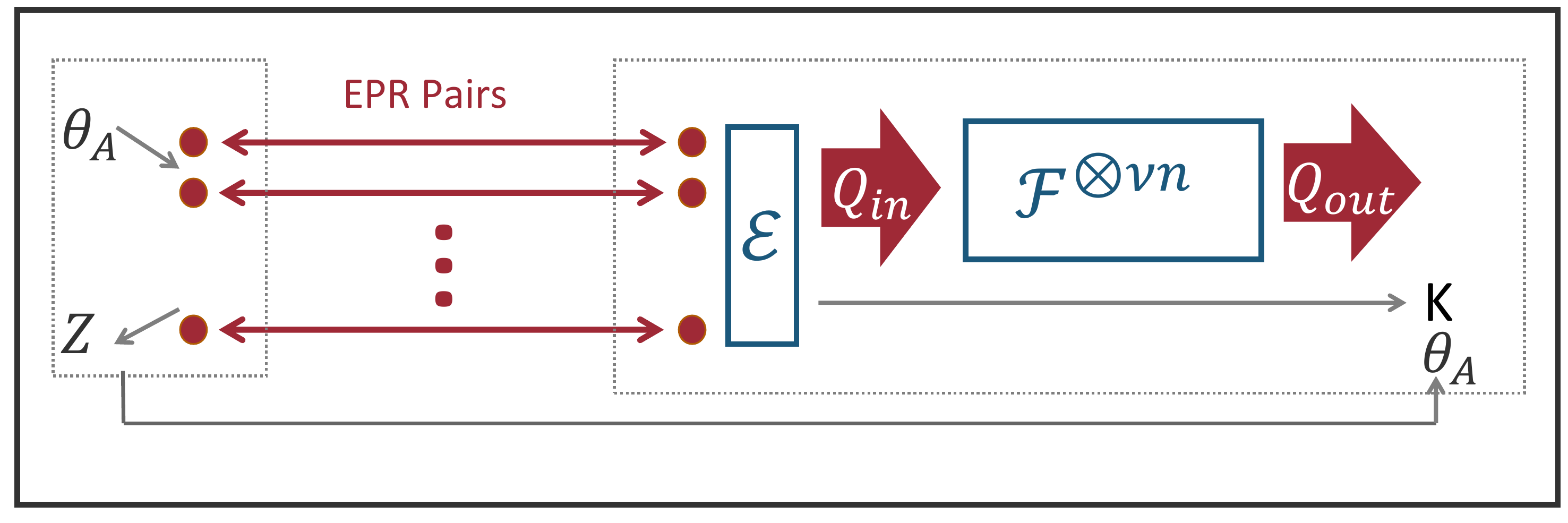}\caption{\label{fig:Memory} 
The scenario of a dishonest Bob. The memory attack is modeled by an encoding $\cE$ that maps (conditioned on some classical outcome $K$) the $n$ modes to the memory input $Q_{\inn}$. The memory $\cM$ is modeled by $\nu n$ uses of the channel $\cF$.  We consider the situations where the encoding $\cE$ is arbitrary, a mixture of Gaussian channels or independent and identical over a small numbers of signals $m$.
}\end{center}\end{figure}

Our goal is to quantify the relation between the capacity of Bob's quantum memory and the security for Alice. 
The latter is obtained if one of the strings, let's say $s_0$ for simplicity, is uniformly random and uncorrelated with $B'$. 
Fortunately, a standard method known in cryptography as privacy amplification can be used to ensure the desired property even if the adversary still holds significant amount
of information about $Z_0$: if the length $\ell$ of the hashed string $s_0$ is roughly equal to the conditional smooth min-entropy $H^{\epsilon_1}_{\min}(Z_0|B')$ of the input of the hash function $Z_0$ given Bob's information $B'$, i.e., $\ell =  H^{\epsilon_1}_{\min}(Z_0|B')-2\log1/(\epsilon_A-4\epsilon_1 )$~\cite{renner05}, then $s_0$ is $\epsilon_A$-close in trace distance to 
uniform and uncorrelated from $B'$.
For $\epsilon_1=0$, the smooth min-entropy is equal to the negative logarithm of the success probability that Bob can guess $Z_0$ by measuring $B'$, and its generalization for $\epsilon\geq 0$ is obtained by maximizing it over all states that are $\epsilon$-close~\cite{renner05,koenig-2008} (see Appendix~\ref{app:SmoothEntropies} for details). 

We thus have to lower bound $H^{\epsilon_1}_{\min}(Z_0|B')$, for which we follow ideas developed for the discrete variable case~\cite{Koenig2012,schaffner2010}.  
First we use an inequality from~\cite{Koenig2012} to bound $H^\epsilon_{\min}(Z_0|B')$ by means of the classical capacity of Bob's memory channel. Note that Bob's information $B'$ at the end of the protocol is given by $B' =\cM(Q_\inn)B_\cl $ with classical information $B_\cl = K\theta_AW_0W_1$. Then, if we denote by $\cP^\cM_{\suc}(k)$ the optimal success probability to reliably send $k$ classical bits through $\cM$, we have that~\cite{Koenig2012}
\begin{equation}\label{eq:psucc}
2^{-H^{\epsilon_1}_{\min}(Z_0|B')} \leq     \cP^\cM_{\suc}(\lfloor H^{\epsilon_2}_{\min}(Z_0|B_\cl)- \log1/(\epsilon_1-\epsilon_2)^2  \rfloor ) \, .
\end{equation}
This inequality reduces the problem to bounding the smooth min-entropy conditioned on Bob's classical information only.

An additional problem arises because we do not know which of the two strings $Z_0,Z_1$ Bob does not learn. This problem can be solved by using the min-entropy splitting theorem~\cite{wulli}, which says that if the uncertainty about the whole string is high, i.e.  $H^{\epsilon_2}_{\min}(Z|B_\cl)\geq \lambda$, then so it is in average for the sub-strings $Z_0$ and $Z_1$. More precisely, there exists a random variable $D$ such that $H^{\epsilon_2}_{\min}(Z_D|B_\cl D)\geq \lambda/2-1$. Note that this result relies crucially on the fact that the information upon which one conditions, i.e., $B_\cl D$ is classical and not quantum~\cite{splitting}. 
Finally, we can remove the dependence on the error-correction information $W_0,W_1$  by simply subtracting the maximal information contained in  $W_0,W_1$, i.e., the number of bits $\ell_\ec = \log |W_0W_1|$.

In conclusion, we obtain security for Alice if $\ell =  - 1/2 \log  \cP^\cM_{\suc}( n r_\ot  )-\log 1/(\epsilon_A-4\epsilon_1)$ where 
\begin{equation}
r_\ot = 1/2\left(  \lambda^{\epsilon_2} (n)-\frac{\ell_\ec}{n}\right) - \frac 1n\left(2\log\frac1{(\epsilon_1-\epsilon_2)} - 1\right) \, ,
\end{equation}
and $\epsilon_A > 4\epsilon_1 > \epsilon_2 \geq 0$. 
Here,  $\lambda^\epsilon (n)$ stands for a lower bound on the smooth min-entropy rate 
\begin{equation}\label{eq:URrate}
\frac 1n H_{\min}^{\epsilon}(Z|\theta_A  K) \geq \lambda^{\epsilon}(n) \, . 
\end{equation} 

We can now relate Alice's security to the classical capacity $C_\cl (\cF)$ of $\cF$ whenever the success probability of reliably sending classical information through $\cF$ at a rate $R$ higher than $C_\cl (\cF)$ decays exponentially  $\cP^{\cF^{\otimes n}}_{\suc}( nR) \leq 2^{-n\xi (R-C_\cl (\cF))}$~\footnote{In information theory this is referred to as a strong converse for the classical capacity and has been shown for many channels.} (see Discussion for examples). Then, by a simple calculation we find that security for Bob can be obtained for large enough $n$ if the condition
\begin{equation}\label{eq:Cond}
r_\ot(n) - \nu C_\cl (\cF) > 0 
\end{equation}
is satisfied. Moreover, the length of $s_0,s_1$ can be chosen as $\ell = n \xi (r_\ot(n) - C_\cl (\cF) ) -\cO(\log1/\epsilon_A)$. 

In order to analyse the security we have to find tight lower bounds $\lambda^\epsilon (n)$ for the inequality~\eqref{eq:URrate}, which is a special kind of uncertainty relation. The other important quantity is the EC rate which in practice is directly determined by the protocol. For the following discussions we use the standard formula $\ell_\ec/n = H(X^A_\delta)-\beta I(X^A_\delta :X^B_\delta)$, where $X^{A}_\delta$ and $X^{B}_\delta$ are the random variables induced if both players are measuring $X_\delta$. The parameter $\beta\leq 1$ is called the efficiency of the EC protocol and values of about $0.96$ are practical using currently available codes~\cite{jouguet2011,jouguet2014,gehring2014}.

\section{Bit commitment in the noisy storage model}
A bit commitment protocol consists of a commitment phase where Bob commits to a bit $c$, and an open phase where Alice learns $c$. Honest Alice wants to be ensured that Bob cannot change his commitment $c$ after completion of the commitment phase (binding). And Bob wants that Alice cannot learn $c$ before the open phase (hiding). Similarly to the OT protocol, we use composable security definitions (described in Appendix~\ref{app:BCSecDef}) using security parameters $\epsilon_C,\epsilon_H,\epsilon_B$ for the correctness, hiding and binding conditions. 

The protocol for commitment is similar to the quantum part of OT and consists of steps (Q1)-(Q3) except that in (Q2) Bob is measuring all his signals in the basis corresponding to the bit $c$ he wants to commit to $\theta_B=c,\ldots,c$. The open phase is purely classical~\cite{damgaard2008}:   
\begin{itemize}
\item[(BC1)] Bob sends $c$ and $Y$ to Alice who defines the substring $Z_I,Y_I$ of $Z,Y$ containing the elements $I=\{i | \theta^i_A = c\}$. She then accepts if $Y_I\in B_{\epsilon_C} (Z_I)$ and rejects otherwise.    
\end{itemize}
Here, the set $B_{\epsilon}(Z_I)$ is the $\epsilon$-typical set of outcomes for Bob if both measure in the same basis and Alice obtains outcome $Z_I$. If $n$ is large enough, Alice will accept with probability $\epsilon_C$ due to the property of typical sequences. Moreover, the hiding condition is satisfied perfectly since Bob does not send any information in the commitment phase.  

As in OT, it is evident that Bob can cheat perfectly if he has a perfect quantum memory. But under the same assumptions on Bob's memory as discussed for the OT protocol, we show in Appendix~\ref{app:BCanalysis} that the hiding condition is satisfied for sufficiently large $n$ if 
\begin{equation}\label{eq:CondBC}
1/2 ( \xi \lambda^\epsilon (n) - \log V_{\epsilon_C}/n)  - \nu  C_\cl(\cF)  >  0 \, ,
\end{equation}
where $V_{\epsilon_C} = \max_Y |B_{\epsilon_C}^{-1}(Y)|$ with $B_\varepsilon^{-1}(Y)=\{Z \, | Y\in B_\varepsilon(Z) \}$ the set of Alice's outcomes $Z$ for which $Y$ would be accepted. Note that in our case, $|B_\epsilon^{-1}(Y)|$ is independent of $Y$ and the maximization can be omitted.  In particular, if the transmissivity of the channel between Alice and Bob is $\tau$, one finds that $\log V_\epsilon = n H(X^B_{\sqrt{\tau}\delta}|X^A_\delta)$, which is approximately the reverse reconciliation rate $ H(Y_{\delta}|Z_\delta)$ increased by $\log1/\sqrt{\tau}$ (see Appendix~\ref{app:BCanalysis} for details).  

\begin{figure}\begin{center}\includegraphics*[width=8cm]{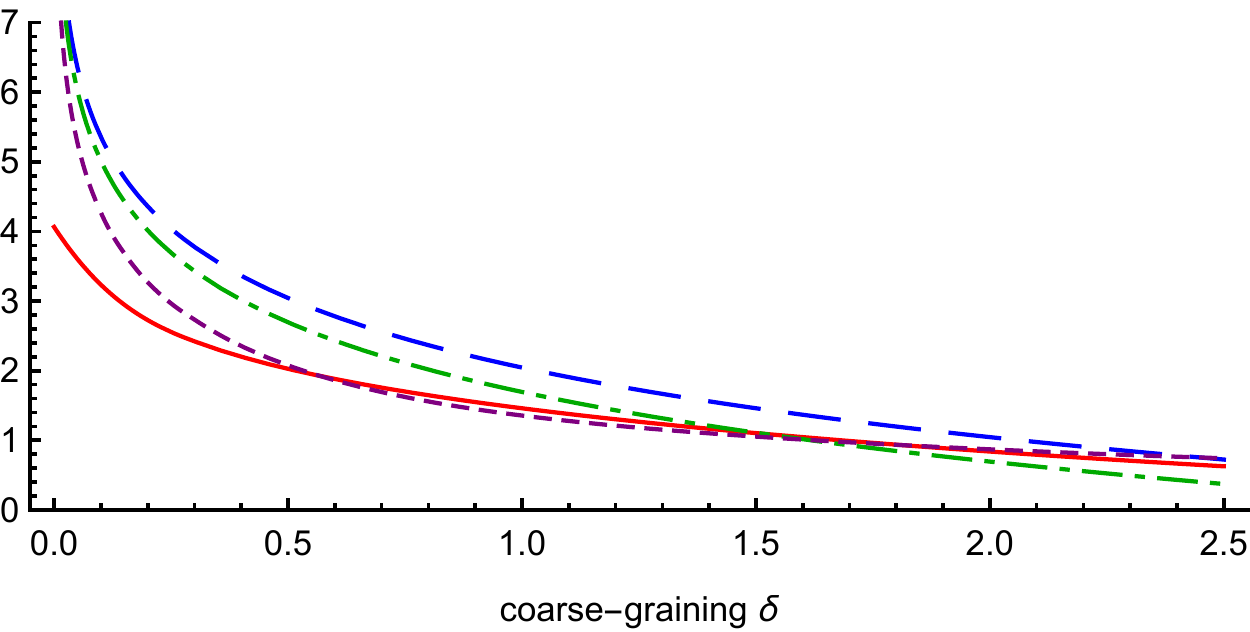}\caption{\label{fig:UR} 
One can achieve security in the noisy storage model without assumptions if the 
error-correction rate $\ell/n$ (short dashed) is below the entropy rate $\lambda^\epsilon_{\text{Maj}} $ (straight), under Gaussian encoding assumption if  $\ell/n$ is below $\lambda^\epsilon_{\gauss} $ (long dashed), and under independent and identical encodings over  $m=25$ signals if $\ell/n$ is below $\lambda^\epsilon_{\iid} $ (dashed-dotted). The horizontal axis describes the coarse graining $\delta$ of the homodyne detection of the $X$ and $P$ quadrature. We chose practical protocol parameters of $n=10^8$ signals and security parameters $\epsilon_A=\epsilon_B=10^{-9}$. The error-correction rate is plotted for an efficiency of $\beta=0.96$ and correlations obtained by an EPR state with squeezing of $10.8$dB and one-sided losses of $5\%$. We further emphasize that the situation looks similar for the BC protocol for the considered parameters. 
}\end{center}\end{figure}

 \section{CV uncertainty relations for the noisy storage model}
Both the security of OT and BC rely on tight bounds $\lambda^\epsilon$ in~\eqref{eq:URrate}. To obtain such bounds using~\eqref{eq:psucc} we first need an uncertainty relation.
In order to derive it, let us grant additional power to a dishonest Bob who can prepare an ensemble of states 
$\{\rho^k\}_k$ according to $K$ himself and send it to Alice who performs on any mode randomly $X_\delta$ or $ P_\delta$. After completing all measurements, Alice sends her basis choice $\theta_A$ to Bob who has to guess the outcomes $Z$ of the measurements. As Bob does not know the measurement choice prior to preparation, the uncertainty principle forbids Bob to know $Z$ perfectly, i.e., there is no state for which the outcomes of both $X$ and $P$ are certain. As the uncertainty principle holds independently for any state of the ensemble, it is intuitively clear that $K$ cannot bring any advantage, which is why we omit it in the following (see e.g.~\cite{Nelly12} for details). 

Traditionally, the uncertainty principle for $X$ and $P$ is captured as a lower bound on the product of their standard deviations $\sigma_X\sigma_P \geq \hbar/ 2$~\cite{kennard1927}. Here, we are interested in an entropic version for coarse grained measurements. For the sake of illustration, let us first consider the equivalent relation for the well-known Shannon entropy, i.e., $H(X) = - \sum_x p_x \log p_x$ if $X$ is distributed according to $\{p_x\}$. For $n=1$, the conditional Shannon entropy of interest is given by $H(Z|\theta_A)= 1/2( H(X_\delta ) + H(P_\delta) )$, where $H(X_\delta)$ denotes the Shannon entropy of the outcomes of the $X_\delta$ measurement and similar for $P_\delta$. However, it is known that $H(X_\delta) + H(P_\delta) \geq -\log c(\delta ) $ with $c(\delta) = \delta^2/(\pi e \hbar)$~\cite{Birula84}, and thus, $H(Z|\theta_A) \geq -1/2 \log c(\delta)$. This inequality can straightforwardly be generalized to $n\geq 1$.   

Let us now  turn to the uncertainty relation of interest in~\eqref{eq:URrate} expressed by the smooth min-entropy instead of the von Neumann entropy. Such uncertainty relations have previously been analysed for maximally complementary qubit measurements~\cite{damgaard2007,Nelly12}. Here, we analyze this uncertainty relation for the first time for $X$ and $P$ measurements. In strong contrast to the situation of the Shannon entropy discussed before, a non-trivial relation is only possible for coarse-grained $X$ and $P$ measurement but not for continuous ones (see Appendix~\ref{sec:uncertaintypreliminaries}). Hence, the lower bound has to be derived for $X_\delta$ and $P_\delta$ directly, which makes it very challenging. 

\begin{figure}\begin{center}\includegraphics*[width=8.5cm]{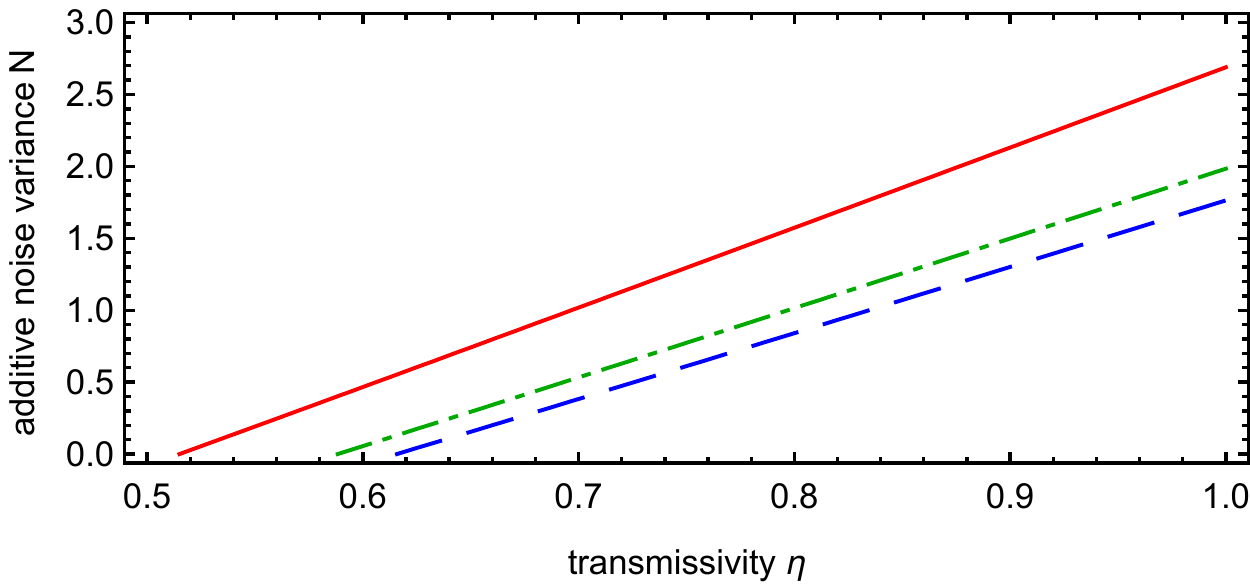}\caption{\label{fig:SecurityGauss} 
The left side of the plots correspond to secure regions for the OT protocol under Gaussian encoding assumption ($\lambda_{\text{Gauss}}^\epsilon $), depending on the transmissivity $\eta$ (horizontal axis) and the additive noise variance $N$ (vertical axis) of malicious Bob's Gaussian memory channel. The thermal noise of the memory is set to $N_\ther =0$ and the maximal photon constraint to $N_{\max}=30$. The different curves correspond to different numbers of quantum memories $\nu n$, squeezing strengths $s$ and one-sided transmissivity $\tau$ of the EPR state for values $(\nu,s,\tau) = (1/8,10.8,0.95)$ (solid), $(1/10,10.8,0.90)$ (dashed) and $(1/12,7.8,0.9)$ (dashed-dotted). We further set the coarse graining to $\delta = 0.2$ (largest optimal value) and choose the other parameters as in Fig.~\ref{fig:UR}. 
}\end{center}\end{figure}

We derive three different uncertainty bounds denoted by $\lambda^\epsilon_{\maj}$, $\lambda_{\gauss}^{\epsilon}$ and $\lambda_{\iid}^\epsilon$. The first $\lambda^\epsilon_{\maj}$ is valid without restrictions on the states $\rho^k$. The derivation is based on a result by Landau and Pollak~\cite{Landau61} that for any two fixed intervals $I,J\in\mathbb R$, the probability $q[I]$ to measure X in $I$,  and $p[J]$ to measure $P$ in $J$ have to satisfy the constraint $\cos^{-1}\sqrt{q[I]} + \cos^{-1}\sqrt{p[J]} \geq \cos^{-1}\gamma$ where $\gamma$ is a function of the length of the intervals of $I$ and $J$. This relation yields an infinite number of constraints for the probability distributions for $X_\delta$ and $P_\delta$. However, it is difficult to derive a lower bound by including all constraints. The bound $\lambda^\epsilon_{\maj}$ is obtained by a relaxation using the majorization technique from~\cite{rudnicki2015} and can be computed recursively (see Appendix~\ref{app:MajUR} for details). The drawback is that the relaxation is not optimal leading to a lose bound leaving the problem open to find a tighter relaxation.  

This problem is overcome in the other two bounds at the expense of additional assumptions. The bound $\lambda_{\gauss}^{\epsilon}$ holds under the assumption that the state is a mixture of Gaussian states. Applied to the noisy-storage model, this assumption requires that the encoding operation $\cE$ (not the quantum memory $\cF$) has to be a mixture of Gaussian operations, which can be justified since non-Gaussian operations are still very challenging in practice and more importantly, their implementations are generally not deterministic.  The explicit form is $ \lambda_\gauss^\epsilon(n)  = \sup ( B^\alpha_{\gauss}(\delta) - {1}/({n(\alpha-1)}) \log {2}/{\epsilon^2} )$, where  
\begin{equation}
B^\alpha_{\text{Gauss}}(\delta) = \frac{1}{1-\alpha} \log\frac
12( 1 + ({1}/{\alpha}) \left( {\delta^2}/({\pi\hbar})\right)^{(\alpha-1)} )
\end{equation}
and the optimization is over all $\alpha>1$. The bound is derived using continuous approximations and the details are in Appendix~\ref{app:GaussianUR}. 

The last bound $\lambda_{\iid}^{\epsilon}$ applies under the assumption that the ensemble states are independent and identical over only $m$-mode states ($m \ll n$). This assumption requires that Bob's encoding operation $\cE$ acts in an identical and independent way on only $m$ modes, which in the language of QKD corresponds to collective attacks. It is a reasonable assumption since coherent operations of all the $n$ modes may require a lot of resources. Moreover, if Bob wants to coherently act on all $n$ modes he has to store all the incoming modes until all $n$ modes arrived, which in principle already requires a short-time memory. It might also be possible for Alice to send the modes with delays in order to enhance her security.  The bound is derived via a reduction of the smooth min-entropy for independent and identical distributions to the Shannon entropy~\cite{renner05,Furrer10} and given by $\lambda_{\iid}^{\epsilon} = - 1/2\log c(\delta)  - \cO(m^2 \sqrt{{m}/{n}} )$ (see Appendix~\ref{URrateIID}). 

The three bounds are shown in Fig.~\ref{fig:UR}. We see that for large $n$, $\lambda_{\gauss}^{\epsilon}$ and $\lambda_{\iid}^\epsilon$ are approximately equal. In fact, in the limit $n\rightarrow \infty$ we find that both converge to the optimal bound determined by the Shannon entropy $-1/2\log c(\delta)$. We further plotted the EC rate for the OT protocol since a necessary condition for security is that $ \lambda^\epsilon (n)$ has to be larger. We assume an EPR state with variance $V=3\hbar$ (about $10.8$dB) and transmission losses on Bob's mode of $0.05$, an error-correction efficiency of $0.95$ and excess noise $0.0005\hbar$. We see that  $\lambda^\epsilon_{\text{Maj}} $ does not provide a very tight bound on the entropy rate such that security can only be achieved under very restrictive assumptions on the quantum memory of a malicious Bob. 

\section{Security for realistic memory devices}
To obtain explicit parameters from~\eqref{eq:psucc}, the second ingredient is a strong converse for the memory in question.
Here, we consider the security conditions for OT~\eqref{eq:Cond} and BC~\eqref{eq:CondBC} for a specific class of memory channels, namely, bosonic thermal-loss channels with additive Gaussian noise. The question of the classical capacity of such channels has only recently been completely solved after settling the minimal output entropy conjecture~\cite{giovannetti2013A,giovannettiB}. Moreover, the exponential decay of successful transmittance above the classical capacity has been established for $\xi=1$  under a maximal photon constraint $N_{\max}$~\cite{wilde2014}, i.e., every code word has (up to negligible probability) maximal $N_{\max}$ photons. In Fig.~\ref{fig:SecurityGauss}, we show when equality in the condition for OT~\eqref{eq:Cond} is attained under Gaussian restriction of a malicious Bob's encoding map $\cE$ (i.e., $\lambda^\epsilon_{\gauss}$) depending on the transmissivity $\eta$ and the additive noise variance $N$ and different fractions of quantum memories $\nu$. Since there is only a week dependence for low thermal noise variance $V_\ther \leq 0.1$, we set it equal to $0$. We see that for $\nu\leq 1/8$ security can already been obtained for memory channels with small losses and additive noise. Note that since $\xi=1$ and that $\log V_{\epsilon_C}$ is approximately $\ell_\ec$ (see Fig.~\ref{fig:UR}), the same result holds for the BC protocol. The approximately same curves are obtained if one restricts the memory attacks to independent and identical encodings over $m\leq 10$ signals, i.e., exchanging $\lambda^\epsilon_{\gauss}$ by $\lambda^\epsilon_{\iid}$ in~\eqref{eq:Cond}. Moreover, security without any restriction on the encoding (i.e., $\lambda^\epsilon_{\maj}$) can be obtained if the fraction of the quantum memory is about $\nu\approx 5\cdot 10^{-3}$. We finally note that independent of the memory model, security for OT and BC can only be obtained with squeezed states and if the transmissivity between Alice and Bob is larger than $1/2$. 

\section{Conclusion}
We have presented a protocol for OT and BC using optical CV systems that provide security in the noisy-storage model. The protocol is practical and uses similar resources as CV QKD. As a key ingredient, we analyze and derive uncertainty relations for CV systems, that can be used along similar lines to analyze the security in the noisy-storage model for other two-party protocols such as secure password-based identification~\cite{Koenig2012,wcsl2010,schaffner2010}. We leave as open problem the task of finding optimal uncertainty relations without any further assumptions. It is possible that such relations can be obtained by linking security again to the quantum capacity of the storage device~\cite{Berta2012,Dupuis2015}, requiring however more sophisticated techniques. Such a result would also pose a challenge to find an explicit strong converse for the quantum capacity of bosonic channels. 

 {\bf Acknowledgements } We would like to thank Anthony Leverrier, Lo\"ick Magnin and Fr\'ed\'eric Grosshans for useful discussions about the continuous-variable world. FF is supported by the Japan Society for the Promotion of Science (JSPS) by KAKENHI grant No. 24-02793. CS is supported by a 7th framework EU SIQS and a NWO VIDI grant. SW is supported by STW Netherlands, as well as an NWO VIDI grant.


\appendix 

\section{CV uncertainty relations for the smooth min-entropy}   \label{app:CVUR}

\subsection{ Smooth min-entropy} \label{app:SmoothEntropies}

Let us consider a classical random variable $Z$ with values in a discrete but possibly infinite set $\cZ$ that is correlated to a quantum system $B$ modeled by a Hilbert space $\cH_B$. If $Z$ is distributed according to $p_z$, the situation can be conveniently be described by the state 
\begin{equation}
\rho_{ZB} = \sum_z p_z \ketbra z z \otimes \rho_B^z \, ,
\end{equation}
where $\rho_B^z$ is the state of system $B$ conditioned on $z\in\cZ$ and $\ket z$ an orthonormal basis. The conditional min-entropy of $Z$ given $B$ is then defined~\cite{koenig-2008} as minus the logarithm of the maximal success probability to correctly infer $Z$ given access to $B$, that is, 
\begin{equation}\label{minEnt}
H_{\min}(Z|B)_\rho = -\log \left( \sup_{\{E_z\}} \sum_z p_z \tr(E_z \rho_B^z) \right) \, ,
\end{equation}
with the supremum taken over all positive operator valued measures (POVM) $\{E_z\}$, i.e., $E_z\geq 0$ and $\sum_z E_z =\idty$. If system $B$ is a classical random variable $Y$ jointly distributed according to $p(z,y)$, the min-entropy is defined for the state $\rho_{ZY} = \sum_{z,y} p(z,y) \ketbra z z \otimes \ketbra y y $ with $\ket y$ an orthonormal basis. 

The smooth min-entropy is then defined as the maximization of the min-entropy over states that are $\epsilon$-close in the purified distance $\cP(\rho,\sigma) = \sqrt{1-F(\rho,\sigma)}$~\cite{Tomamichel09}. Here $F(\rho,\sigma)=(\tr \vert\sqrt{\rho}\sqrt{\sigma}\vert)^2 $ denotes the fidelity. In formulas, this means that 
\begin{align}\label{eq:SmoothMin}
H_{\min}^\epsilon(Z|B)_\rho = \sup_{\tilde\rho_{ZB}} H_{\min}(Z|B)_{\tilde\rho} \, , 
\end{align}
where $\cP(\rho_{ZB},\tilde\rho_{ZB}) \leq \epsilon$. 

The smooth min-entropy satisfies many entopy-like properties. For instance, we will frequently use the following chain rule
\begin{equation}\label{app:ChainRule}
H_{\min}^\epsilon(A|BZ) \geq H_{\min}^\epsilon(A|B) - \log |Z| \, , 
\end{equation}
where $A,B$ are arbitrary systems (quantum or classical) and $Z$ a classical system of dimension $|Z|$. 
For further properties of the smooth min-entropy, we refer the reader to~\cite{tomamichel:thesis} in the finite-dimensional and to~\cite{Furrer10,berta2011} in the infinite-dimensional case.

\paragraph*{Reduction to R{'e}nyi Entropy.} 
We are interested in lower bounding the smooth min-entropy 
\begin{equation}\label{app:URrate}
\frac 1 n H^\epsilon_{\min}(Z|\theta K) \geq \lambda^\epsilon
\end{equation}
 in which the side-information $K$ is classical as well. Note that we omit the subscript $A$ in Alice's measurement choice $\theta$. Because of the maximization in the definition of the smooth min-entropy~\eqref{eq:SmoothMin}, it is very difficult to bound the smooth min-entropy directly. Instead, it is easier to use that it can be related to the conditional $\alpha$-R{\'e}nyi entropies 
\begin{equation}
H_\alpha(A|B)_\rho = \frac 1 {1-\alpha} \log \tr [\rho_{AB}^\alpha (\id_A\otimes \rho_B)^{1-\alpha}] \, .
\end{equation} 
In particular, it holds for $\alpha\in(1,2]$ and any two random variables $X$ and $Y$ that~\cite{Tomamichel08}
\begin{equation}
H^\epsilon_{\min}(X|Y) \geq H_\alpha(X|Y) - \frac{1}{\alpha-1} \log\frac{2}{\epsilon^2} \,  ,
\end{equation} 
We present in Lemma~\ref{lem:MinToReny} a simple generalization of the result to unbounded classical variables $X$ and $Y$. Hence, if we find a bound on the R{\'e}nyi-entropies 
\begin{equation}\label{app:Renyirate}
\frac 1 n H_{\alpha}(Z|\theta K) \geq B^\alpha 
\end{equation}
we obtain a lower bound on the smooth min-entropy with 
\begin{equation} \label{app:URrateRenyi}
\lambda^\epsilon = \sup_{1< \alpha\leq 2} \left(  B^\alpha  - \frac{1}{n(\alpha-1)} \log\frac{2}{\epsilon^2} \right)  \, .
\end{equation}

Moreover, as shown in~\cite{Nelly12}, it suffices to find a bound $H_\alpha(Z|\theta) \geq B^\alpha$ for $n=1$ and without $K$, as such a bound implies that $H_\alpha(X|\theta K)\geq n \lambda$ for strings $X$ and $\Theta$ of any length $n$. This implication basically follows from the fact that the conditional R{\'e}nyi entropies can be expanded as 
\begin{equation}\label{app:RenyExp}
2^{(1-\alpha)H_\alpha(Z|Y) }= \sum_y 2^{(1-\alpha)H_\alpha(X|Y=y)} \, ,
\end{equation} 
where $ H_\alpha(X|Y=y)$ denotes the $\alpha$-R{\'e}nyi entropy of $X$ given $Y=y$.

\subsection{Uncertainty relation for R{\'e}nyi entropy} 

\subsubsection{Preliminaries} \label{sec:uncertaintypreliminaries}
According to the discussion in the previous section it is sufficient to consider the case of $n=1$ and trivial
$K$. That is, the system $A$ is a position-momentum system and
$\theta\in \{0,1\}$ is a uniformly distributed random variable indicating the measurement choice, i.e., $\theta = 0$ and $\theta =1$ corresponding to $X_{\delta_x}$ and $P_{\delta_p}$. For the sake of generality, we allow for different binning $\delta_x$ and $\delta_p$ for the coarse-grained $X$ and $P$ measurement. 

In the following we assume that $\{I_k\}_{k\in \mathbb N}$ and $\{J_k\}_{k\in \mathbb N}$ are two partitions of $\mathbb R$ into intervals of constant length $\delta x$ and $\delta p$, respectively. We denote the probability to measure $X$ ($P$) in interval $I_k$ ($J_l$) by $q_k$ ($p_l$). Using the definition of the R{\'e}nyi entropy together with the expansion in~\eqref{app:RenyExp}, we find that 
 \begin{align}\label{eq:ExpDiscUR}
2^{(1-\alpha)H_\alpha(X|\theta)} = \frac 12 (\sum_k q_k^\alpha + \sum_l p_l^\alpha )  \, . 
\end{align}
Given that $\alpha>1$, an upper bound on the above sums results in a lower bound on $H_\alpha(X|\theta)$.

For simplicity, let us first consider the continuous case for $\delta_x,\delta_p \rightarrow 0$. In this case, the sums in~\eqref{eq:ExpDiscUR} become integrals and we obtain
\begin{align}\label{eq:diffUR}
2^{(1-\alpha)h_\alpha(X|\theta)} = \frac 12 \left(\int q(x)^\alpha \dd x + \int p(x)^\alpha \dd x \right)  \, ,
\end{align}
where $q$ and $p$ are the probability distributions corresponding to $X$ and $P$, and $h_\alpha(X|\theta)$ denotes the differential conditional R{\'e}nyi entropy. While $q,p$ are integrable functions, that is not necessarily true for $q^\alpha,p^\alpha$. Hence, it is possible that \eqref{eq:diffUR} diverges when optimizing over all possible states. So, no uncertainty relation can be shown for continuous measurements. 

The right hand side of~\eqref{eq:diffUR} can be made arbitrarily large even for Gaussian states. For a normal distribution with standard deviation $\sigma$, the differential $\alpha$-R{\'e}nyi entropy is 
\begin{equation}\label{app:DiffRenyi}
h_\alpha(X) = \log \left[ \sqrt{2 \pi} \sigma \alpha^{ \frac{1}{2(\alpha -1)} }\right]  \, .
\end{equation} 
Hence, we obtain for a Gaussian state with standard deviation $\sigma_X$ and $\sigma_P$ in $X$ and $P$ that 
\begin{align}\label{eq:diffURGauss}
2^{(1-\alpha)h_\alpha(X|\theta)_\rho} = (2\pi)^{(1-\alpha)/2} \alpha^{-1/2} \left( \frac{1}{\sigma_Q^{\alpha -1}} + \frac{1}{\sigma_P^{\alpha -1}} \right)  \, . 
\end{align}
This quantity can be made arbitrarily large by taking a sufficiently small standard deviation for either $\sigma_Q$ or $\sigma_P$. However, the divergence is not a problem for coarse-grained outcomes as both of the sums in~\eqref{eq:ExpDiscUR} are upper bounded by $1$.

In order to bound~\eqref{eq:ExpDiscUR}, we have to use that not all possible distributions $q_k$ and $p_l$ are possible since they origin from measurements of the complementary observables $X$ and $P$ that are related via Fourier transform. This relation has been made rigorous by Landau and Pollak~\cite{Landau61} (see also~\cite[Section 2.9]{DymMcKean}), who show that the probability $q[I]$ to measure $X$ in an interval $I$ ($a=|I|$) and the probability $p[J]$ to measure $P$ in an interval $J$ ($b=|J|$) have to
satisfy the inequality 
\begin{equation}\label{eq:LandauPollak}
\cos^{-1}\sqrt{q[I]} + \cos^{-1}\sqrt{p[J]} \geq \cos^{-1}\sqrt{\gamma(a,b)} 
\end{equation}
where 
\begin{equation}
\gamma(a,b) := \frac{ab}{2\pi\hbar} S_0^{(1)}\left(1,\frac{ab}{4\hbar}\right)^2 
\end{equation}
with $S_0^{(1)}$ the 0th radial prolate spheroidal wave function of the first kind. For $ab$ sufficiently small $\gamma(a,b)\approx ab/(2\pi\hbar)$. 

The condition~\eqref{eq:LandauPollak} can be reformulated in the following way~\cite{DymMcKean}: 
\begin{itemize}
	\item[i)] If $0\leq  q[I] \leq \gamma(a,b)$, then all values for $p[J]$ are possible, and
	\item[ii)] if $\gamma(a,b)\leq q[I]$, then $p[J]\leq g(q[I],a,b)$ for 
	\begin{equation}\label{eq:gfunction} 
g(q,a,b):= \left( \sqrt{q \gamma(a,b)} + \sqrt { (1-q) (1-\gamma(a,b) )}\right)^2 \, . 
	\end{equation}
\end{itemize}

This reformulation yields an infinite number of constraints for the probability distributions. Let us assume that $\{q_k\}_k$ and $\{p_l\}_l$ are decreasingly ordered. Then for all $M,N \in \mathbb{N}$, $\{q_k\}_k$ 
and $\{p_l\}_l$ 
have to satisfy 
the constraints 
\begin{align}\label{eq:ConstrSet3}
 	\sum_{j=1}^N  p_{k_j} \leq  g\left(\sum_{i=1}^M q_i , M\delta_x, N\delta_p \right) \, . 
\end{align}
However, it is non-trivial to turn these constraints into an explicit and tight upper bound for~\eqref{eq:ExpDiscUR}. In the following we discuss a particular way that connects the above constraints with a majorization approach, which leads to $\lambda_{\maj}^\epsilon$.

 \subsubsection{Majorization uncertainty relation} \label{app:MajUR}
 
This bound follows from an idea in~\cite{rudnicki2015}. Let us denote by $r$ the decreasingly ordered sequence of both probabilities $\{q_k\}$ and $\{p_l\}$. Then, the expression~\eqref{eq:ExpDiscUR} can be written as 
\begin{equation}
2^{(1-\alpha)H_\alpha(X|\theta)} = \frac 12 \sum_j r_j^\alpha \, . 
\end{equation}
Since the function $r\mapsto \sum_j r_j^\alpha$ is Schur convex, we get an upper bound on~\eqref{eq:ExpDiscUR} if we find a sequence $w$ which majorizes any physically possible sequence $r$. Such a $w$ can be constructed in the following way~\cite{rudnicki2015}.

First note that according to property ii) , $ q[I]+p[J] \leq q[I] + g(q[I],a,b)$, which optimized over all $0\leq q[I] \leq 1$ is equal to $1+\sqrt{\gamma(a,b)}$. Hence, we obtain the constraint
\begin{equation}
q[I]+p[J] \leq 1+ \sqrt{\gamma(a,b)} \, . 
\end{equation}
This constraint implies further that
\begin{equation}
\sum_{j=1}^n r_j \leq  1 + F_n(\delta q,\delta p) \, , 
\end{equation} 
where 
\begin{equation}
 F_n(\delta_x,\delta_p) = \max_{1\leq k\leq n} \sqrt{c\left(k\delta_x,(n-k) \delta_p\right)} \, .
\end{equation}
Note that in the case $\delta_x = \delta_p$ the maximum is attained for $k=\lfloor \frac n 2 \rfloor$.

We can construct a majorizing sequence $w$ by setting recursively
\begin{align}\label{eq:MajSequ}
w_1 = 1 , \  \text{and} \ w_k = F_k - w_{k-1} \ \text{for} \ k\geq 2 \, . 
\end{align}
The obtained bound on the R{\'e}nyi-entropy is 
\begin{equation}\label{eq:MajUR}
H_\alpha(X|\theta) \geq B^\alpha_{\text{Maj}} \, 
\end{equation} 
where 
\begin{equation}
B^\alpha_{\text{Maj}} = \frac{1}{1-\alpha} \log \left( \frac 12 \sum_k w_k^\alpha\right)  \, . 
\end{equation} 
By using~\eqref{app:URrateRenyi}, we obtain the following bound on the smooth min-entropy 
\begin{equation}
 \lambda_\maj^\epsilon := \sup_{1< \alpha\leq 2} \left( B^\alpha_{\text{Maj}} - \frac{1}{n(\alpha-1)} \log\frac{2}{\epsilon^2} \right)  \, .
\end{equation}

Since $B^\alpha_{\text{Maj}}$ depends on the recursively defined sequence $w$ in~\eqref{eq:MajSequ}, there is no closed form and it can only be computed numerically. 
However, one can easily check that $B^\alpha$ is monotonically increasing in $\alpha$, which simplifies the optimization over $\alpha$ required for the calculation of $\lambda_{\maj}^\epsilon$.

\subsubsection{Uncertainty relation for Gaussian states} \label{app:GaussianUR}

In order to obtain a tighter bound, we consider an uncertainty relation that holds for arbitrary Gaussian states or any mixture thereof. In fact, for our application it is important to allow arbitrary and even continuous mixtures of Gaussian states. The reason
is that a coarse-grained quadrature measurement with finite binning on one mode of a
multi-mode Gaussian state results in a continuous mixture of
Gaussian states in the remaining modes (and not in a Gaussian state
itself). Since conditioning on part of the measurement outcomes on Alice's mode is needed to generalize the uncertainty relation from
$n=1$ to $n>1$~\cite{Nelly12}, this level of generality is crucial. 

\begin{thm}\label{thm:BoundGauss}
Let $\alpha \in (1,2]$. For an arbitrary and possibly continuous convex combination of Gaussian states and coarse-grained measurements $X_{\delta_x}$ and $P_{\delta_p}$ holds that 
\begin{equation}\label{eq,thm:BoundGauss}
H_\alpha(Z|\theta) \geq  B^\alpha_{\text{Gauss}}(\delta)  \, ,
\end{equation}
where
\begin{equation}
B^\alpha_{\text{Gauss}}(\delta_x,\delta_p) = \frac{1}{1-\alpha} \log\frac
12\left( 1 + \frac{1}{\alpha} \left( \frac{\delta_x\delta_x}{\pi\hbar}\right)^{(\alpha-1)} \right)
\end{equation}
\end{thm} 

\begin{proof}
We first assume that the state is a Gaussian state. Let us recall that we have to upper bound the sums in~\eqref{eq:ExpDiscUR}. Denoting the probability density function of the continuous $X$ measurement by $q(x)$, a simple application of Jensen's inequality results in 
\begin{align}\label{eq:Jensen}
q_k^\alpha = \left(\int_{I_k} q(x) dx \right)^\alpha \leq \delta_x^{\alpha-1} \int_{I_k} q(x)^\alpha dx \, .
\end{align}
Hence, using the formula for the R{\'e}nyi entropy of a Gaussian state~\eqref{app:DiffRenyi}, we find for a state with standard deviation $\sigma_X$ for the $X$ measurement that 
\begin{align}
\sum_k q_k^\alpha \leq \delta^{\alpha -1} \int q(x)^\alpha dx = g(\tilde\sigma_X) \, , 
\end{align}
where $\tilde\sigma_X=\sigma_X/\delta_X$ is the relative standard deviation and 
\begin{align}
g(x)  = \frac{1}{\sqrt{\alpha} (\sqrt{2\pi}x)^{\alpha-1} } \, . 
\end{align}
Note that the bound $g(\tilde\sigma_X)$ becomes very bad if $\tilde\sigma_X$ is very small. In particular, it can exceed the trivial upper bound on $\sum_k q_k^\alpha$ given by $1$. We avoid that problem by simply bounding
\begin{align}
\sum_k q_k^\alpha \leq \min \{  g(\tilde\sigma_X) , 1\} \, . 
\end{align}

Let us use that the standard deviations of the $X$ and $P$ distribution satisfy $\sigma_X\sigma_P \geq \hbar/2$, which translates into $\tilde\sigma_X \tilde\sigma_P\geq \hbar/(2\delta_x\delta_p)$ for the relative standard deviations. For the following it is convenient to define $\tilde\hbar := \hbar/(\delta_x\delta_p)$. Given that we want to maximize the quantity over all Gaussian states, we can without loss of generality assume that $\tilde\sigma_X\geq \tilde\hbar/\sqrt{2}$ and that $\tilde\sigma_P = \tilde\hbar/(2\tilde\sigma_X)\leq \tilde\sigma_X$. A straightforward calculation then results in 
\begin{align*}
 &\sum_k q_k^\alpha + \sum_k p_k ^\alpha \leq \min \{g(\tilde\sigma_X), 1\}  + \min \{ g(\tilde\sigma_P) , 1\} \\
 &\qquad \leq  \left\{ \begin{array}{ll} 
 g(\tilde\sigma_X) +g(\tilde\sigma_P)  & \mbox{if }\ \tilde \sigma_X\leq \tilde\hbar \sqrt{\frac{\pi}{2}} \alpha^{\frac 1{2(\alpha -1)}} \, ,
 \\ 
 1+ g(\tilde\sigma_X) &\text{otherwise.} 
 \end{array} \right.
\end{align*}
One can check easily that $g(\tilde\sigma_X) +g(\tilde\hbar/(2\tilde\sigma_X))$ is monotonically increasing in $\tilde \sigma_X$, and $1+g(\tilde\sigma_X)$ is monotonically decreasing in $\tilde\sigma_X$. Hence, the maximum of the right hand side is attained exactly for $  \tilde \sigma_X = \tilde\hbar \sqrt{\frac{\pi}{2}} \alpha^{\frac 1{2(\alpha -1)}}$. Plugging this value in, we obtain that 
\begin{align}
 \sum_k q_k^\alpha + \sum_k p_k ^\alpha   \leq  1 + \frac{1}{\alpha} \left(\frac{1}{\pi} \right)^{{\alpha-1}} \left( \frac{(\delta_x \delta_p)}{\hbar}\right)^{(\alpha-1)}\,  , 
\end{align}
which finishes the proof for Gaussian states. 

Let us now assume that the state is given by $\rho = \int_Y d\mu(y) p(y) \rho^y$ with $(Y,d\mu)$ a sigma-finite measure space, $p$ a probability distribution over $Y$ and $\rho^y$ a Gaussian state for any $y$. It then follows that the $X$ measurement maps $\rho$ to an element of $L^1(\mathbb R)$ that can be written as $\rho_Q = \int_Y d\mu(y) p(y) \rho_Q^y$ with $\rho_Q^y$ the Gaussian distribution of the position of $\rho^y$. The same holds for the $P$ measurement. It thus follows that 
\begin{align}
& \sum_k \left(\int_{I_k}d x \int_Y d \mu(y) p(y) \rho_Q^y(x) \right)^\alpha \\ 
 =&  \sum_k \left( \int_Y d \mu(y) p(y)\int_{I_k}d x \rho_Q^y(x) \right)^\alpha \\ 
 \leq&   \sum_k  \int_Y d \mu(y) p(y)\left(\int_{I_k}d x \rho_Q^y(x) \right)^\alpha \\   
 = & \int_Y d \mu(y) p(y)  \sum_k  \left(\int_{I_k}d x \rho_Q^y(x) \right)^\alpha \, ,
\end{align}
where the two equalities follow from Fubini's theorem (since all integrals and sums are bounded) and the inequality from the convexity of the function $x\mapsto x^\alpha$ on the non-negative reals (for $\alpha \in (1,2]$). Thus, by the linearity of the integral we obtain the desired result. 
\end{proof}

Similar to the majorization uncertainty relation, we get a bound on the smooth min-entropy via~\eqref{app:URrateRenyi}
\begin{equation}\label{eq:URrateMaj}
\frac 1n H_{\min}^{\epsilon}(Z|\theta K) \geq \lambda_\gauss^\epsilon(\delta_x,\delta_p,n)  \, ,
\end{equation}
with 
\begin{equation}
 \lambda_\gauss^\epsilon(\delta_x,\delta_p,n)  := \sup_{\alpha} \left( B^\alpha_{\text{Gauss}}(\delta_x,\delta_p) - \frac{1}{n(\alpha-1)} \log\frac{2}{\epsilon^2} \right)  \, .
\end{equation}

Let us show that the performance of the inequality in the asymptotic limit is optimal in the sense that it converges to the bound obtained for the Shannon entropy. In order to do so, we consider the limit 
\begin{align} 
\lim_{\epsilon\rightarrow 0}\lim_{n\rightarrow \infty} \lambda_\gauss^\epsilon(\delta_x,\delta_p,n) = \lim_{\alpha \rightarrow 1}   B^\alpha_{\text{Gauss}}(\delta_x,\delta_p) \, ,
\end{align}  
where we used that $ B^\alpha_{\text{Gauss}}(\delta_x,\delta_p)$ is independent of $n$ and monotonically decreasing in $\alpha$. A straightforward calculation yields that 
\begin{align} \label{app:AsymptoticURrate}
\lim_{\alpha \rightarrow 1}   B^\alpha_{\text{Gauss}}(\delta_x,\delta_p) = - \frac 12 \log c(\delta_x,\delta_p) \, ,
\end{align}  
where $c(\delta_x,\delta_p):= {\delta_x\delta_p}/({\pi e \hbar}) $ is the uncertainty bound for the Shannon entropy, i.e., $H(X|\theta) \geq -(1/2)\log c(\delta_x,\delta_p)$~\cite{Birula84}. 

Intuitively, the reason for this convergence is that we use Jensen inequality~\eqref{eq:Jensen} to bound the discrete entropy to the differential entropy, i.e., $H_\alpha(X_\delta) \geq h_\alpha(X)- \log \delta$, together with the fact that the uncertainty relation for differential Shannon entropies~\cite{Birula75,Beckner75} becomes an equality for any pure Gaussian state. Hence, for the limit case $\alpha\rightarrow 1$, we simply obtain 
\begin{align}
2 H(X_\delta|\theta) & =  H(Q_\delta) + H(P_\delta) \\
& \geq h(Q) + h(P) -\log \delta^2 \\
& \geq -\log \pi e \hbar - \log\delta^2 \\
& = -\log c(\delta) \, ,
\end{align}
where we used $\delta_Q=\delta_P=\delta$ for simplicity.

\subsection{ Uncertainty relation under iid assumption} \label{sec:URIID}

The following bound is based on the property that the smooth min-entropy of $n$ independent  and identically distributed (iid) random variables converges to the Shannon entropy in the asymptotic limit~\cite{renner05}. We use a result derived in~\cite{Tomamichel08}, which for iid random variables $X^n$ and $Y^n$ reads as
\begin{equation}\label{app:AEP}
\frac 1n H_{\min}^{\epsilon}(X^n|Y^n) \geq H(X|Y) - \frac{4}{\sqrt{n}} \log(\eta(X)_\rho)^2 \sqrt{\log\frac 2 {\epsilon^2}} \, 
\end{equation}
where $\eta(X) =  2+ 2^{H_{1/2}(X)_\rho/2}$. The above result holds also for random variables over infinite alphabets if $H(X)<\infty$~\cite{Furrer10}. The crucial point for the application of the inequality in~\eqref{app:AEP} is that the correction term beside the Shannon entropy is independent of the conditioning variable $Y$. 

Let us assume that Bob produces an ensemble of $n$-mode states according to an independent and identical distribution (iid) over only $m$ modes such that the state on $A$ and $K$ has the form $\rho_{A^nK^n}=(\sigma_{A^mK^m})^{\otimes n/m }$, where we assume that $n/m\in\mathbb N$. Then, also the random variable $Z^n$ obtained by measuring randomly either $X_{\delta_x}$ or $P_{\delta_p}$ has the same structure. Applying the inequality~\eqref{app:AEP}, we obtain that 
\begin{align}
\frac 1n H_{\min}^{\epsilon}(Z^n|\theta^n K^n) & \geq \frac 1m H(Z^m|\theta^m K^m)\\
 & \quad  - 4\sqrt{\frac{m}{n}} \log(\eta(Z^m))^2 \sqrt{\log\frac 2 {\epsilon^2}} \, .
\end{align}
At this point, we can simply use the uncertainty relation for the Shannon entropy~\cite{Birula84} 
\begin{align}
H(Z^m|K^m\Theta =\theta ) + H(Z^m|K\Theta=\bar\theta) \geq - m\log c(\delta_x,\delta_p) \, ,
\end{align}
where $\bar\theta=(1-\theta_i)_{i=1}^m$ denotes the complementary basis choice of $\theta=(\theta_i)_{i=1}^m$. This inequality implies that 
\begin{align*}
& H(Z^m|\Theta K) \\
 = & \frac{1}{2^m} \sum_{\theta} \frac 12\left( H(Z^m|K\Theta =\theta ) + H(Z^m|K\Theta=\bar\theta) \right)  \\
  \geq & - \frac{m}{2}\log c(\delta_x.\delta_p) \, .
\end{align*} 
Hence, we obtain the uncertainty relation 
\begin{equation}
\frac 1n H_{\min}^{\epsilon}(Z^n|\theta^n K^n) \geq \lambda^\epsilon_\iid (\delta,m,n)\ ,
\end{equation}
where 
\begin{align}\label{URrateIID}
\lambda^\epsilon_\iid (\delta,m,n) &  = -\frac 12\log c(\delta_x,\delta_p)  \\
& \quad - 4\sqrt{\frac{m}{n}} \log(\eta(Z^m))^2 \sqrt{\log\frac 2 {\epsilon^2}} \, .
\end{align}
Note that even though the right-hand side still depends on the distribution of $Z^m$, it is not conditioned on $K$ and Alice can estimate it. Particularly, in the application to oblivious transfer or bit commitment, we can assume that Alice distributes the average ensemble state, and thus, knows the distribution over $Z$ by herself. Note further that $\log(\eta(Z^m)) =\cO(m)$ such that 
\begin{equation}\label{app:IIDUR} 
\lambda^\epsilon_\iid (\delta,m,n) = -\frac 12\log c(\delta_x,\delta_p)  - \cO(m^2 \sqrt{\frac{m}{n}} ) \, .
\end{equation}
Similarly to when we restricted to Gaussian states, we find that in the asymptotic limit, the bound converges to $-1/2\log c(\delta_x,\delta_p)$.


\section{Oblivious Transfer}  \label{app:OT}

\subsection{Composable security definitions} \label{app:OTSecDef}

In the following, we denote random variables by capital letters, e.g, $S_0,S_1$ for Alice's output.  The uniform distribution of a random variable $X$ is denoted by $\tau_X$ and the classically maximally correlated state of two random variables $X$ and $Y$ with same range by $\Omega_{XY}$, i.e., $\Omega_X=\tau_X$, $\Omega_Y=\tau_Y$, and $\Omega_{X|Y=y}=\delta_{x,y}$. Moreover, we set $[n]=\{1,2,...,n\}$ and $\bar x= 1-x $ for any binary variable $x$.

We use the composable security definitions from~\cite{Koenig2012}. 
\begin{de} \label{def:OT}
A protocol between two parties Alice and Bob that takes input $T$ in $\{0,1\}$ from Bob and outputs on Alice's side two bit strings $S_0,S_1$ in $\{0,1\}^\ell$ and on Bob's side $\tilde S$ in $\{0,1\}^\ell$ is called an $(\epsilon_C,\epsilon_A,\epsilon_B)$-secure (sender-randomized) $\textrm{OT}^\ell$ protocol if the following conditions hold: 
\begin{itemize}
	\item  The protocol is $\epsilon_C$-correct. That is, if both parties follow the protocol, then the output of the protocol $\rho_{S_0S_1 \tilde{S} T}$ satisfies for $t \in \{0,1\}$ 
	\begin{equation}\label{eq:Corr}
	\Vert \rho_{S_0S_1\tilde S|T=t} - \tau_{S_{\bar t}} \otimes \Omega_{S_t\tilde S} \Vert_1 \leq \epsilon_C \, . 
\end{equation}	 
	\item The protocol is $\epsilon_A$-secure for Alice. That is, if Alice follows the protocol, then for any strategy of Bob with output $\rho_{S_0S_1B'}$, where $B'$ denotes Bob's register at the end of the protocol, there exists a random variable $D$ with range $\{0,1\}$ such that  
	\begin{equation} \label{eq:SecA}
   \Vert \rho_{S_{\bar{D}}S_D DB'} - \tau_{S_{\bar{D}}} \otimes \rho_{S_D DB'} \Vert_1 \leq \epsilon_A \, .
	\end{equation}
	\item The protocol is $\epsilon_B$-secure for Bob. That is, if Bob follows the protocol, then for any strategy of Alice with output $A'$, resulting in the joint output state $\rho_{A'\tilde{S}T}$, there exist random variables $S_0',S_1'$ such that $\rho_{A'S_0'S_1'\tilde S T}$ satisfies $\pr{\tilde S\neq S_T'} \leq \epsilon_B$ and  
		\begin{equation} \label{eq:SecB2}
   \Vert\rho_{A'S_0'S_1'|T=0} - \rho_{A'S_0'S_1'|T=1}  \Vert_1 \leq \epsilon_B \, .
	\end{equation}
\end{itemize}
\end{de}

\subsection{Security analysis for oblivious transfer} \label{app:OTanalysis}

The conditions for correctness are that $\tilde S $ is with high probability equal to $S_t$ and that $S_0S_1$ are uniformly distributed. The first condition relies on the error-correction protocol. We assume in the following that the procedure manages to correct the error with probability $\epsilon_{\ec}$. The second condition follows from the security definition for Alice~\eqref{eq:SecA} by using the monotonicity of the trace norm. Hence, if security for Alice holds with $\epsilon_A$, correctness holds with at least $\epsilon_C = \epsilon_\ec +2\epsilon_A$.

Security for Bob holds since he only sends the sets $I_0,I_1$ during the entire protocol, which due to the random choice of the measurement by Bob are independent of $t$.  This has been made precise in~\cite{damgaard2007}. 

Let us consider security for Alice. Recall that Bob's memory attack is given by an encoding $\cE$ mapping the $n$ modes to systems $Q_\inn$ and $K$, where $Q_\inn$ is the input of his memory channel $\cM=\cF^{\otimes \nu n}$  and $K$ some additional classical information. Hence, after completing the entire protocol, Bob's system is given by $B'=Q_\out B_\cl$, where $Q_\out = \cM(Q_\inn)$ and $B_{\cl}=\theta_A KWHC$ all his classical information. Here,  $W=(W_0,W_1)$ denotes the error-correction information, and $H=(F_0,F_1)$ the $2$-universal 
hash functions used for privacy amplification.

According to~\eqref{eq:SecA}, we have to show that there exists a random variable $D$ such that 
\begin{equation}\label{eq:SecA2}
   \Vert \rho_{S_{D}S_{\bar D} D Q_\out B_{\cl}} - \tau_{S_D} \otimes \rho_{S_{\bar{D}} D Q_\out B_{\cl}} \Vert_1 \leq \epsilon_A. 
\end{equation}
The privacy amplification lemma~\cite{renner05,Tomamichel10} against infinite-dimensional quantum adversaries~\cite{berta2011} tells us that~\eqref{eq:SecA2} is satisfied for
\begin{equation}\label{eq:hashLength1}
\ell \geq H^{\epsilon_1}_{\min}(Z_D|S_{\bar D} D Q_\out B_{\cl}) - 2 \log\frac{1}{\epsilon_A-4\epsilon_1} \, ,
\end{equation}  
with $\epsilon_1\geq 0$ arbitrary such that $\epsilon_A\geq 4 \epsilon_1$. 

Hence, it remains to find a tight lower bound on the smooth min-entropy $H^{\epsilon_1}_{\min}(Z_D|S_{\bar D} D Q_\out B_{\cl})$. For this purpose we follow similar arguments as in~\cite{Koenig2012,schaffner2010}. Therein, a central ingredient is a bound of the smooth min-entropy $H_{\min}(U|\cM(Q_\inn) V)$ with $U,V$ classical and $\cM$ a quantum channel by the success probability to send classical information at a rate $R$ through $\cM$ 
\begin{equation}
\cP_{\suc}^{\cM}(nR):= \sup_{\rho_k , D_k}  \frac{1}{2^{nR}} \sum_{k} \tr(D_k\cM(\rho_k)) \, ,
\end{equation}
where the supremum runs over ensembles of code states $(\rho_k)_{k=1}^{nR}$ and POVM's $(D_k)_{k=1}^{nR}$ acting as a decoder. It has been shown in~\cite{Koenig2012} that (see also Lemma~\ref{lem:QMtoCap})
\begin{equation}\label{app:MinBoundPsucc}
H^{\epsilon +\epsilon'}_{\min}(U|\cM(Q_{\inn})V) \geq -\log \cP^\cM_{\suc}\left(\lfloor H^{\epsilon}_{\min}(U|V)_\rho - \log \frac 1{\epsilon'^2} \rfloor\right) \, .
\end{equation} 

Applying the chain rule~\eqref{app:ChainRule}, we first bound
\begin{align*} 
 H^{\epsilon_1}_{\min}(Z_D| S_{\bar{D}} D Q_\out B_{\cl})  \geq H^{\epsilon_1}_{\min}(Z_D| D Q_\out B_{\cl}) - \ell  \, ,
\end{align*}
where we used that $\log |S_{\bar D}| = \ell$. The smooth min-entropy $H^{\epsilon_1}_{\min}(Z_D| D Q_\out B_{\cl})$ on the right hand side can then be lower bounded by means of inequality~\eqref{eq:hashLength1} by  
\begin{align*} 
 -\log\left( \cP^\cM_{\suc}\left(\lfloor H^{\epsilon_2}_{\min}(Z_D| D  B_{\cl})  - \log\frac 1{(\epsilon_1-\epsilon_2)^2} \rfloor\right)\right) \, . 
\end{align*}
Plugging the bounds in~\eqref{eq:hashLength1} and solving for $\ell$, one easily finds that~\eqref{eq:SecA2} is satisfied if we choose $\ell$ smaller or equal to 
\begin{align*}
& -\frac 12  \log \left( \cP^\cM_{\suc}\left(\lfloor H^{\epsilon_2}_{\min}(Z_D| D  B_{\cl})  - \log\frac 1{(\epsilon_1-\epsilon_2)^2} \rfloor\right)\right)   \\
& -  \log\frac{1}{\epsilon_A-4\epsilon_1}  \, .
\end{align*}

The goal of the next part is to lower bound the smooth min-entropy $ H^{\epsilon_2}_{\min}(Z_D| D  B_{\cl})$. For that lower bound, we use the min-entropy splitting theorem~\cite{damgaard2007}, (see also Lemma~\ref{lem:MinSplit}), which tells us that there exists a random variable $D$ such that 
\begin{equation}
H_{\min}^{\epsilon} (Z_D|DB_{\cl}) \geq \frac 12 H_{\min}^{\epsilon} (Z_0Z_1|B_{\cl}) - 1   \, .
\end{equation} 
Given that Bob's classical register $B_{\cl}$ is given by $\theta_A KWHC$, we finally get via~\eqref{app:ChainRule} that 
\begin{align*}
H_{\min}^{\epsilon} (Z_0Z_1|B_{\cl}) & \geq H_{\min}^{\epsilon} (Z_0Z_1|\theta_AK)-\log|W|-\log|C|  \\ 
& \geq  H_{\min}^{\epsilon} (Z_0Z_1|\theta_AK)- \ell_{\ec} - 2\log\frac 1{\epsilon_C} \, ,   
\end{align*}
where $\ell_\ec=\log|W|$, and we used that the hash functions are drawn independently at random. 

Concluding the above discussion, we arrive at the following bound on the length of the string that enables security for Alice. 
\begin{thm}\label{thm:secA}
Let us assume that $H^\epsilon_{\min}(Z|K \theta) \geq  n \lambda_\epsilon(n)$ and Bob's memory channel is given by $\cM$. Then the protocol for OT consisting of steps (Q1)-(Q3) followed by (OT1)-(OT3) is $\epsilon_A$-secure for Alice if 
\begin{align} \label{thm:eq,SecA}
\ell =  -\frac 12  \log  \cP^\cM_{\suc}\left( \lfloor  n r_\ot   \rfloor\right)  -  \log\frac{1}{\epsilon_A-4\epsilon_1} \, ,
\end{align} 
where 
\begin{equation}
r_\ot  :=   \lambda_{\epsilon_2}(n) - r_{\ec} - \frac 1 n \left( 2 \log\frac 1{\epsilon_C} - \log\frac 1{(\epsilon_1-\epsilon_2)^2} -1\right) \, 
\end{equation}
and $\epsilon_1,\epsilon_2\geq 0$ arbitrary such that $\epsilon_A> 4\epsilon_1> 4\epsilon_2$.  
\end{thm}
Note that when the right-hand side of~\eqref{thm:eq,SecA} is negative, a secure implementation of the OT protocol is not possible.

Let us consider the case where Bob's quantum channel $\cF$ is such that the success probability to send classical bits above the classical capacity $\cC_\cl(\cF)$  decreases exponentially
\begin{equation}\label{eq:StrongConvAss}
 \cP^{\cF^{\otimes n}}_{\suc}(nR)\leq 2^{-n\xi(R-\cC_\cl(\cF))} \,  .
\end{equation}
This property is generally referred to as \emph{strong converse} for the classical capacity. Note that for channels $\cF$ for which it is only known that property~\eqref{eq:StrongConvAss} is satisfied for a rate $C_{\text{sc}} >\cC_\cl(\cF)$ usually referred to as a strong converse capacity of $\cF$, the following discussion holds similarly with $\cC_\cl(\cF)$ replaced by $C_{\text{sc}}$. 

A simple calculation shows that if Bob's memory is of the form $\cM=\cF^{\otimes\nu n}$ and~\eqref{eq:StrongConvAss} holds, then the condition 
\begin{equation}`
r_\ot  > \nu \cC_\cl(\cF) \, 
\end{equation}
is sufficient to obtain security for a large enough $n$. Moreover the length of the strings $s_0,s_1$ can be chosen as 
\begin{equation}
\ell = n \xi (r_\ot - \cC_\cl(\cF)) \, .
\end{equation}
A necessary condition for security is thus 
\begin{equation}\label{eq:URminusEC}
1/2(\lambda^{\epsilon} - r_{\ec}) > 0 \ . 
\end{equation}

Let us analyse the above condition in the asymptotic limit. We know according to~\eqref{app:AsymptoticURrate} and \eqref{app:IIDUR} that $\lambda_{\gauss}^\epsilon$ and $\lambda^\epsilon_{\iid}$ converge to $-1/2\log c(\delta)$, where we assume in the following for simplicity that $\delta_x=\delta_p=\delta$. Using the exponential deFinetti theorem or the post-selection technique applied to CV protocols~\cite{circac09,2013DeF}, it is easy to convince oneself that this bound holds in the asymptotic limit without any assumptions (e.g., Gauss or iid). This insight yields the asymptotic formula $\lambda_{\asym}(\delta) := -1/2\log c(\delta)$. 

The error-correction rate in the asymptotic scenario is given according to the Slepian-Wolf theorem~\cite{Slepian71} as the conditional Shannon entropy $H(X^A_\delta|X^B_\delta)$, where $X^A_\delta$ ($X^B_\delta$) is the outcome of Alice's (Bob's) coarse-grained $X$ measurement. We assume here that the state is symmetric with respect to $X$ and $P$. This assumption is reasonable because if both parties are honest, the state is an EPR state with one-sided loss. Let us denote the conditional variance of $X^A$ given $X^B$ by $V_{A|B}$~\footnote{If the covariance matrix of two Gaussian random variables $X$ and $Y$ is denoted by $\Gamma_{XY}$, then $V_{X|Y}= \det \Gamma_{XY}/V_Y$.}. Then, if $\delta \ll
 V_{A|B}$, we find with good approximation that $r_{\ec}= H(X^A_\delta|X^B_\delta) \approx h(X^A|X^B) - \log\delta$, where  $h(X^A|X^B)$ is the conditional differential Shannon entropy of $X^A$ given $X^B$. Hence, we obtain
\begin{align*}
\lambda_\asym - r_{\ec} & \approx \log \sqrt{e\pi} - h(X^A|X^B) \\ 
& = \log \sqrt{e\pi\hbar} - \log\left(\sqrt{2\pi e V_{A|B}}\right)  \\ 
& = \log \left(\sqrt{\frac{\hbar}{2V_{A|B}}}  \right) \, .
\end{align*}  

In order to satisfy~\eqref{eq:URminusEC}, we need that $V_{A|B}<\sqrt{\hbar/2}$. Given a Gaussian state with covariance matrix as in~\eqref{eq:EPR} (that is, an EPR state with one-sided losses $1-\tau$  and $\xi=0$), the condition above reads as ($\hbar = 2$) 
\begin{equation}
V_{A|B} = \frac{(1-\tau)V+\tau} {(1-\tau) + \tau V} < 1 \,  .
\end{equation}
Since $V\geq 1 $, the condition can only be satisfied if the transmissivity $\tau > 1/2$ and for a non-trivial squeezing $V>1$.

\section{Bit commitment} \label{app:BC}

\subsection{Security definitions} \label{app:BCSecDef}

Let us first introduce the notation. In the bit commitment phase, Bob inputs a bit $C$ to which he commits. In the open phase Alice outputs a bit $\tilde C$ and a flag $F$, where $F\in \{\text{accept},\text{reject}\}$ depending whether Alice accepts or rejects the commitment. 

We use composable security definitions adapted from~\cite{Koenig2012}. 
\begin{de}
A protocol between two parties Alice and Bob that consists of a commitment phase where Bob commits to a bit $C$ and an open phase in which Alice outputs $\tilde C$ and a flag $F\in \{\text{accept},\text{reject}\}$ is called an $(\epsilon_C,\epsilon_H,\epsilon_B)$-secure bit commitment protocol if the following conditions hold:  
\begin{itemize}
\item The protocol is $\epsilon_C$-correct. That is, if both parties are honest, then it holds that $\pr{ \ \tilde C \neq C \ |\ F=\text{accept} \ } \leq \epsilon_C$ and $ \pr{ F=\text{reject}} \leq \epsilon_C$. 
\item The protocol is $\epsilon_H$-hiding. That is, if Bob is honest then for any strategy of Alice with joint output state $\rho_{A'C}$, it holds after the commitment phase that
\begin{equation} 
\Vert \rho_{A'|C=0} - \rho_{A'|C=1} \Vert_1 \leq \epsilon_H \, . 
\end{equation}
\item The protocol is $\epsilon_B$-binding. That is, if Alice is honest, then for any strategy of Bob, there exists after the commitment phase a random variable $D$ in $\{0,1\}$ such that for any value $D'$ that Bob wants to convince Alice to accept, it holds that   
\begin{equation}\label{app:Binding}
\pr{ D' \neq D \ | F= \ accept \ } \leq \epsilon_B \, . 
\end{equation}
\end{itemize}
\end{de}

\subsection{Security analysis for bit commitment} \label{app:BCanalysis}

It is easy to see that the BC protocol is correct. The first condition $\pr{ \ \tilde C \neq C \ |\ F=\text{accept} \ } \leq \epsilon_C$ is satisfied due to the definition of the protocol. The second condition holds for sufficiently large $n$ due to the properties of typical sets. The protocol is perfectly hiding as Bob does not send any information to Alice during the commitment phase. So, the interesting case is to show that the protocol is binding as long as Bob's quantum memory satisfies some constraints. We start with a lemma. 

\begin{lem}\label{lem:GuessingSet}
Let $\rho_{XU}$ be an arbitrary state on $X$ and $U$, where $X$ is a classical system with alphabet $\cX$ and $U$ is arbitrary (possibly quantum). Moreover, let $B(x)\subset \cX$ for $x\in\cX$. Then for all $\epsilon\geq 0$, the optimal probability to correctly guess $Y$ in $B(X)$ given the system $U$ is upper bounded by 
\begin{equation}
\sup_{\cE}\PP_\rho[ \cE(U) \in B(X) ] \leq \max_{y}|B^{-1}(y)| 2^{-H^{\epsilon}_{\min}(X|U)} \, + 2\epsilon ,
\end{equation} 
where $B^{-1}(y)=  \{x \, | \, y \in B(x)\} $ and the supremum runs over all channels $\cE$ that map $U$ to $\cX$. 
\end{lem}

\begin{proof}
We consider first the case $\epsilon =0$. 
Let $Y=\cE(U)$ for an arbitrary channel $\cE$ with range $\cX$, and $p(x,y)$ the corresponding joint distribution of $X$ and $Y$. We then compute
\begin{align*}
&\PP_\rho[ Y\in B(X)] \\
&= \sum_{x,y} \delta(y\in B(x)) p(x,y)  \\ 
& = \sum_y p(y) \sum_x \delta(y\in B(x))   p(x|y)  \\ 
& \leq \sum_y p(y) \max_{x'} p(x'|y)  \max_{y'}\sum_x \delta(y'\in B(x)) \\
& = \max_{y'}|B^{-1}(y')| \sum_{y} p(y) 2^{-H_{\min}(X|Y=y)} \\
& = \max_{y}|B^{-1}(y)| 2^{-H_{\min}(X|Y)_\rho} \\
& \leq \max_{y}|B^{-1}(y)| 2^{-H_{\min}(X|U)_\rho} \, .
\end{align*}
In the third equality, we used that the min-entropy of a distribution $q(x)$ is $-\log \max_x q(x)$, the forth equality uses a basic property of the classical conditional min-entropy, and the last inequality is due to the data-processing inequality $H_{\min}(X|U) \leq H_{\min}(X|\cE(U))$  (see e.g.~\cite{tomamichel:thesis}). Since the upper bound holds for any $\cE$ and is independent of $\cE$, we can also take the supremum over all $\cE$  concluding the result for $\epsilon=0$.  

In order to generalize the above estimate to $\epsilon >0 $, we take an arbitrary state $\tilde\rho_{XU}$ such that $\cP(\rho_{XU},\tilde\rho_{XU}) \leq \epsilon$. We denote the joint probability distribution obtained by applying an arbitrary strategy $\cE$ on $\rho_{XU}$ and $\tilde\rho_{XU}$ by $p(x,y)$ and $\tilde p(x,y)$, respectively. We then compute that  
\begin{align*}
 \PP_\rho[ Y\in B(X) ]& = \sum_{x,y} \delta(y\in B(x)) p(x,y)  \\ 
& = \sum_{x,y} \delta(y\in B(x)) \tilde p(x,y) \\
& \quad + \sum_{x,y} \delta(y\in B(x)) \left(p(x,y)-\tilde p(x,y)\right)\\
& \leq \max_{y}|B^{-1}(y)| 2^{-H_{\min}(X|U)_{\tilde\rho}} \\
& \quad + \Vert \rho_{XY} - \tilde \rho_{XY}\Vert_1 \,  , 
\end{align*}
where the last inequality follows from the result for $\epsilon =0$ applied to $\tilde\rho_{XU}$, and $\rho_{XY}= \cE(\rho_{XU})$ and similar for $\tilde\rho_{XY}$. Due to the monotonicity of the trace norm under channels, we have that $\Vert \rho_{XY} - \tilde \rho_{XY}\Vert_1 \leq \Vert \rho_{XU} - \tilde \rho_{XU}\Vert_1$. Finally, by using that for any two states $\sigma$ and $\eta$, the purified distance satisfies $\Vert \sigma -\eta\Vert_1 \leq 2 \cP(\sigma,\eta)$, and by taking the minimum over all states $\tilde\rho_{XU}$ with  $\cP(\rho_{XU},\tilde\rho_{XU})\leq \epsilon$, we arrive at the desired inequality. 
\end{proof}

In the following, additional to (Q1)-(Q3) and (BC1), we assume that Alice sets F=reject also if the condition $n/2-\varepsilon' \leq |I|\leq n/2+\varepsilon'$ is violated. Moreover, we define a general verification set $B_{\epsilon_C}^m(Z_I)$, where $m=|Z_I|$ indicates the length of the string and $\epsilon_{C}$ the probability that $Y_I$ is in $B_{\epsilon_C}^m(Z_I)$. We discuss the proper choice of the verification set $B_{\epsilon_C}^m(Z_I)$ as the set of typical subsequences after the proof of the following theorem. 

\begin{thm}\label{thm:Binding}
Let us assume that $H^\epsilon_{\min}(Z|K \theta) \geq  n \lambda_\epsilon(n)$ and Bob's memory channel is given by $\cM$. Then, the BC protocol with verification set $B_{\epsilon_C}^m(z)$, $z\in\cX^m$, is $\epsilon_B$-binding with
\begin{equation}
\epsilon_B \leq   V^{\varepsilon'}_{\epsilon_C}  \cP^\cM_{\suc}\left( \lfloor  \frac n 2 \lambda_{\epsilon_1}(n)  - \log\frac 1{(\epsilon_2-\epsilon_1)^2} -1  \rfloor\right) + 2\epsilon_2 \, 
\end{equation}
where $V_{\epsilon_C}^{\varepsilon'} := \max_{y}|[B^{n/2+\varepsilon'}_{\epsilon_C}]^{-1}(y)|$ and $\epsilon_B/2 > \epsilon_2 > \epsilon_1\geq 0$ arbitrary. 
\end{thm}

\begin{proof} 
We use the same notation as in the proof for OT and denote Bob's system after the commitment phase by $B'=Q_{\out} K$ where $K$ is a classical register and $Q_{\out}=\cM(Q_\inn)$.
Let us denote by $Z_0 $ and $Z_{1}$ the substrings of $Z$ in which Alice chose basis $0$ and $1$, respectively. According to~\eqref{app:Binding}, we have to show that there exists a random variable $D$ such that the probability that Bob convinces Alice that his commitment was $\bar D$ is smaller than $\epsilon_B$. Let us denote Bob's system after the commitment phase by $B'$ and his opening strategy by $\cE$ from $B'$ to $\cZ$. Since Alice accepts only if Bob can answer correctly with a string $Y=\cE(B')$ in $B_{\epsilon_C}(Z_{\bar D})$, the probability can be bounded by Lemma~\ref{lem:GuessingSet} as
\begin{align*}
\PP[ \cE(B') \in B_{\epsilon_C}(X_{\bar D})|F=\text{accept}] & \leq V_{\epsilon_C}^{\varepsilon'} 2^{-H^{\epsilon}_{\min}(X_{\bar D}|B'D)} \\
& \quad + 2\epsilon \, ,
\end{align*} 
where we used that $V_{\epsilon_C}^{\varepsilon '} \geq  \max_{y}|[B_{\epsilon_C}^{|Z_D|}]^{-1}(y)|$. Hence, it remains to find a lower bound on $H^{\epsilon}_{\min}(Z_{\bar D}|B'D)$. For that we follow a similar strategy as in the proof of OT. 

In order to define the random variable $D$, we use the min-entropy splitting theorem (Lemma~\ref{lem:MinSplit}). This lemma tells us that there exists a random variable $D$ such that 
\begin{align}
H^{\epsilon}_{\min}(Z_{\bar D} |DK\theta_A) \geq \frac 12 H^{\epsilon}_{\min}(Z_0Z_1|K\theta_A) -1 \, .
\end{align}
This technique allows us to define a state $\rho_{Z_DKQ_{\out}D}$ such that $\rho_{Z_DKQ_{\out}(D=d)} = \rho_{Z_dKQ_{\out}}$. We then bound the smooth min-entropy of this state by using~\eqref{app:MinBoundPsucc} (see also Lemma~\ref{lem:QMtoCap})
\begin{align*}\label{eq:minSpliBC}
& H^{\epsilon}_{\min}(Z_{\bar D}|K\cM(Q_{\inn})\theta_A D) \\
& \geq    -\log\left( \cP^\cM_{\suc}\left(\lfloor H^{\epsilon_1}_{\min}(Z_{\bar D}| K\theta_A D ) - \log\frac 1{(\epsilon-\epsilon_1)^2} \rfloor\right)\right) \, ,
\end{align*}
which concludes the proof by setting $\epsilon=\epsilon_2$. 
\end{proof}

Let us assume that the memory channel satisfies a strong converse similar to~\eqref{eq:StrongConvAss} given by 
\begin{align}
\cP_{\suc} ^{\cF^{\otimes \nu n}}(nR) \leq 2^{-n\xi (R-\nu C(\cF))}\, .
\end{align}
According to Theorem~\ref{thm:Binding}, we obtain an $\epsilon_B$-binding protocol if $\epsilon_B-2\epsilon_2$ is smaller than 
\begin{align*}
2^{- n \left[ (\xi /2)( \lambda_{\epsilon_1}(n)  - (1/n) \log\frac 1{(\epsilon_1-\epsilon_2)^2} - 1/n ) -\log V_\varepsilon^{\varepsilon'} /n- C(\cF) \right] } \,  . 
\end{align*}
Hence, the necessary condition for obtaining security for sufficiently large $n$ is given by 
\begin{align} \label{eq:BCcond1}
&\frac{\xi}{2}  \lambda_{\epsilon_1}(n)(\delta_x)    - \frac{\log V_\varepsilon^{\varepsilon'}}{n} - \nu C(\cF) \\ &  - \frac 1 n ( \log\frac 1{(\epsilon_B-2\epsilon_2)(\epsilon_1-\epsilon_2)^2} +1 ) > 0  \, , 
\end{align}
where $\epsilon_1,\epsilon_2$ can be chosen arbitrarily according to $\epsilon_B/2 > \epsilon_2 > \epsilon_1\geq 0$.
  
The canonical choice for the verification set  $B^m_{\varepsilon}(z)$ is the set of $\varepsilon$-typical sequences corresponding to the output on Bob's side if Alice's outcome is $z$.  Let us now assume that the state shared by Alice and Bob is given by an EPR state where Bob's mode is sent through a fiber with transmissivity $\tau$. Thus, it is of form~\eqref{eq:EPR}. Then the outcome $X^A$ of a continuous $X$ measurement on Alice's side relates to Bob's outcome $X^B$ of an $X$ measurement by 
\begin{equation}
X^B = \sqrt{\tau} X^A + \cN(V_{B|A}) \, ,
\end{equation}
where $V_{B|A}$ denotes the conditional variance of $X^B$ given $X^A$ (which is independent of $X^A$ for~\eqref{eq:EPR}), and $\cN(V)$ denotes the normal distribution centered at $0$ with variance $V$. It follows that $1/\sqrt{\tau}X^B-X^A$ is distributed according to $\cN(V_{B|A}/\tau)$. Hence, in the case of continuous measurements, the typical set corresponds to the typical set of the normal distribution $\cN(V_{B|A}/\tau)$. Note that the state considered here is symmetric in $X$ and $P$.

Alice and Bob measure coarse-grained versions of $X_A$ and $X_B$. In order to re-scale Bob's outcome directly, it is convenient to choose different discretizations for Alice and Bob given by $\delta_A=\delta$ and $\delta_B= \sqrt{\tau} \delta$. Then Bob's discretized outcome is distributed according to the discretization of $X_A+ \cN(V_{B|A}/\tau)$ with binning $\delta$. Hence, the verification set is translation invariant and given by $B_\varepsilon^m(z) = z+ T^m_{\varepsilon}({V_{B|A}/\tau,\delta})$, where $T_{\varepsilon}^m({V_{B|A}/\tau,\delta})$ denotes the $\varepsilon$-typical sequences of length $m$ sampled according to the discretized normal distribution $ \cN(V_{B|A}/\tau)$ with binning $\delta$. 

We are interested in the inverse set $[B_\varepsilon^m]^{-1}(y)$, which due to the appropriate scaling of Bob's outcome is equal to $B_\varepsilon^m(z)$. Note that the size of the set of $\varepsilon$-typical sequences of length $m$ of a random variable $W$ is upper bounded by $2^{n(H(W)+\varepsilon)}$. Since the distribution of $X^B_{\sqrt{\tau}\delta}$ conditioned on $X^A_\delta$ is independent of the value of $X^A_\delta$, we have that the entropy of the conditional distribution for a fixed value of $X^A_\delta$ is equal to the average over all values of $X^A_\delta$. This then implies that $\log T^m_{\varepsilon}({V_{B|A}/\tau,\delta}) = n(H(X^B_{\sqrt{\tau}\delta}|X^A_\delta)+\varepsilon)$, which yields 
\begin{equation}
V_\varepsilon^{\varepsilon'} = (\frac{n}{2}+\varepsilon') (H(X^B_{\sqrt{\tau}\delta}|X^A_\delta)+\varepsilon)  \, .
\end{equation} 

As in the case of OT, let us consider condition~\eqref{eq:BCcond1} in the asymptotic limit, i.e., taking $\lambda_{\asym}(\delta)$. For simplicity, we again assume that $\sqrt{\tau} \delta \ll V_{A|B}$ such that we can approximate $H(X^B_{\sqrt{\tau}\delta}|X^A_\delta) \approx h(X^B|X^A) - \log\sqrt{\tau}\delta$. Then, we find that in the asymptotic limit the condition for security is given by 
\begin{align*}
\xi \lambda_\asym(\delta)  -  H(Y_{\sqrt{\tau}\delta}|X_\delta) &  \approx \xi \log\sqrt{e\pi} - h(Y|X) + \log\tau \\ 
& \quad  +  (\xi-1)\log\frac 1 \delta \, .
\end{align*}
The last term shows that the value of $\xi$ is very crucial, in the sense that if $\xi>1$, we can increase the value arbitrarily by making $\delta$ small. Unfortunately, in the case of bosonic channels $\xi=1$.  For $\xi=1$ and a state given in~\eqref{eq:EPR}, the condition $\xi \lambda_\asym(\delta)  -  H(Y_{\sqrt{\tau}\delta}|X_\delta) >0$ translates to ($\hbar =2$)
\begin{equation}
\frac{(1-\tau)V + \tau}{\tau V} < 1 
\end{equation}
which is satisfied if $\tau >  V/(2V-1)$. Hence, a non-trivial squeezing is required and $\tau >1/2$ in the limit $V\rightarrow \infty$.

\section{Model of the EPR source} \label{app:EPR}

For the simulations used to generate Fig.~\ref{fig:UR} and~\ref{fig:SecurityGauss}, we assume an EPR state with variance $V=\hbar \cosh 2r$ with $r$ the squeezing parameter (see, e.g.,~\cite{Weedbrook12}). Alice's mode is loss-free and Bob's mode is sent through a fiber with transmissivity $\tau$. We further assume in some cases a non-zero excess noise $\xi$. Then, the covariance matrix of the Gaussian state shared between Alice and Bob is given by  
\begin{equation} \label{eq:EPR}
 \left( \begin{array}{cc}
V I & \sqrt{\tau(V^2-1)} Z \\
\sqrt{\tau(V^2-1)} Z  & V_B(\tau,\xi) I \\
\end{array} \right)
\end{equation} 
with $I$ the identity in $\mathbb C^2$, $Z=\text{diag}(1,-1)$ and $V_B(\tau,\xi)= \tau V +(1-\tau)\hbar/2 + \tau \xi$. Since large distances are not particularly required for the usefulness of OT and BC, we assume in the plots that $\tau=0.94$ and $\xi=0.0005\hbar$. Moreover, we use a variance of $V=3\hbar$ which corresponds to a squeezing strength of about $10.8$dB.

\section{Bosonic Gaussian memory channels} \label{app:GaussianChannels}
 
Let us consider the security conditions for the OT and BC protocol if Bob's memory channel (or a part of it) can be modelled by a phase-insensitive Gaussian channel that acts on a single-mode covariance  matrix as  
\begin{equation} \label{eq:BosChannel}
\Gamma \mapsto T \Gamma T^{T} + N \, ,
\end{equation} 
where $T=\text{diag}(\sqrt{t}, \sqrt{t}) $ and $N=\text{diag}(v,v)$ such that $v \geq 0$ and $v \geq (t -1)$. In the following, we denote the corresponding quantum channels by $\cF_{t,v}$. 

For phase-insensitive Gaussian channels a strong converse has recently been established~\cite{wilde2014,bardhan2014,bardhan2014B}. Note first that the classical capacities for bosonic channels are only bounded under a mean-energy constraint, i.e., if the mean photon number $N_{\av}$ of the average code state is finite. Then, the classical capacities are given by~\cite{giovannetti2013A,giovannettiB} 
\begin{equation}
C(\cF_{t,v}|N_{\av}) = g\left(t N_\av +(t+v-1)/2\right) - g\left(\frac{t+v-1}{2}\right) \, , 
\end{equation}
where $g(x)=(x+1)\log(x+1) - x\log x$. 

For a strong converse bound to hold, a mean-photon number constraint is not sufficient and one has to impose a maximal-photon-number constraint. More precisely, let $\rho^n$ be the average channel input for $n$ channel uses of $\cE_{t,v}$. Then we say that a family of codes $\{\rho^n\}_n$ satisfies a maximal-photon-number constraint (MPNC) with $N_{\max}$ if~\cite{wilde2014}
\begin{equation}\label{eq:MPNC}
\tr\left( \Pi_{nN_{\max}} \rho^n \right) \geq 1-\delta(n)
\end{equation}  
where $\Pi_{nN_{\max}}$ denotes the projector onto the subspace with at most $nN_{\max}$ photons and $\delta(n)$ decays exponentially in $n$. 

The strong converse theorem for any phase-insensitive channel $\cF$ from~\cite{bardhan2014B} then says that the success probability for the transmission under the MPNC decays as  
\begin{equation} 
\cP_{\suc}^{\cF^{\otimes n}}(nR|N_{\max}) \leq 2^{-n(R-C(\cF|N_{\max})-\delta_1 } + 2^{n\delta_2} + \delta_3(n) \, ,  
\end{equation} 
where $\delta_1,\delta_2$ are arbitrary small constants and $\delta_3(n)=\sqrt{\delta(n) + \sqrt{\delta(n)} + \delta_4(n)}$ with $\delta(n)$ given in~\eqref{eq:MPNC} and $\delta_4(n)$ is exponentially decreasing in $n$. Hence, we have a strong converse of the form~\eqref{eq:StrongConvAss} with $\xi =1$, and we can analyse the security conditions for OT and BC given in~\eqref{eq:Cond} and~\eqref{eq:CondBC} in the main text. 

For the plots in the main text we consider the most common phase-insensitive channel given by a thermal-loss channel with additive noise. A thermal-loss channel can be modeled by mixing the mode by a beam splitter with transmissivity $\eta$ with a thermal state with average photon number $N_\ther$. In terms of the parameters $t,v$ in~\eqref{eq:BosChannel}, it is expressed by $t=\eta$ and $v=(1-\eta)(1+2 N_\ther)$. Moreover, if we include additional additive Gaussian noise $V_n$, the parameters are $t=\eta$ and $v=(1-\eta)(1+2N_\ther)+V_n$.

\section{Technical lemmas} \label{app:TechLem}

\begin{lem}\label{lem:MinToReny}
Let $X$ and $Y$ be possibly infinite discrete classical systems. It then holds for any $1<\alpha \leq 2$ that 
\begin{equation}
H^\epsilon_{\min}(X|Y) \geq H_\alpha(X|Y) - \frac{1}{\alpha-1} \log\frac{2}{\epsilon^2} \, .
\end{equation} 
\end{lem}

\begin{proof}
The lemma has been shown for finite-dimensional systems in~\cite{Tomamichel08}. An easy way to show it in the infinite-dimensional case is by means of a finite-dimensional approximation result shown in~\cite{Furrer10}. This approximation allows us to obtain 
\begin{align}
H^\epsilon_{\min}(X|Y)_\rho & \geq H^{\epsilon-\delta}_{\min}(X|Y)_{P_k\rho P_k} \\
& \geq H_\alpha(X|Y)_{P_k\rho P_k} - \frac{1}{\alpha-1} \log\frac{2}{(\epsilon-\delta)^2} \, 
\end{align}
where $P_k= P_k^X \otimes P_k^Y$ is a projector onto a finite-dimensional subspace such that $P_k\rho P_k$ is $\delta$-close to $\rho$, for some $\delta >0$. Note that such a projection always exists for any $\delta$. Next, we use that $H_\alpha(X|Y)_{P_k\rho P_k} \rightarrow H_\alpha(X|Y)_{\rho }$ for $k\rightarrow \infty$. This limit follows simply since all the sums involved in the definition of the $\alpha$ entropy converge absolutely, and thus, can be rearranged. We get as conclusion that 
\begin{align}
H^\epsilon_{\min}(X|Y)_\rho & \geq H_\alpha(X|Y)_{\rho} -  \frac{1}{\alpha-1} \log\frac{2}{(\epsilon-\delta)^2} \, 
\end{align}
holds for any $\delta>0$. And thus in the limit $\delta$ to 0 we obtain the desired result.  
\end{proof}

The following statement has been shown in~\cite{Koenig2012} and generalizes straightforwardly to infinite dimensions using the same strategy as in the proof above based on the approximation theorem in~\cite{Furrer10}. 
\begin{lem} \label{lem:QMtoCap}
Let $\rho_{XKQ_{\inn}}$ be a state of classical random variables $XK$ correlated with a quantum system $Q_{\inn}$ and $\cF$ a quantum channel from $Q_{\inn}$ to $Q_{\out}$. Then, it holds that 
\begin{equation}
H^{\epsilon +\epsilon'}_{\min}(X|\cF(Q_{\inn})K) \geq -\log \cP_{\suc}(k_{\epsilon,\epsilon'}) \, ,
\end{equation} 
where  $k_{\epsilon,\epsilon'} = \lfloor H^{\epsilon}_{\min}(X|K)_\rho - \log1/\epsilon'^2\rfloor $.
\end{lem}

The technique of min-entropy splitting is due to~\cite{wulli}, and used as the following Lemma in~\cite{damgaard2007,schaffner2010}, which generalizes by a simple application of the approximation in~\cite{Furrer10} to arbitrary alphabet sizes. 
\begin{lem}\label{lem:MinSplit}
Let $X_0,X_1,Y$ be classical random variables. Then there exists a random variable $D$ with range $\{0,1\}$ such that 
\begin{equation}
H_{\min}^{\epsilon} (X_D|DY) \geq \frac 12 H_{\min}^{\epsilon} (X_0X_1|Y) - 1   \, .
\end{equation}    
\end{lem}

\bibliography{libraryCVNoisy}

\begin{thebibliography}{62}%
\makeatletter
\providecommand \@ifxundefined [1]{%
 \@ifx{#1\undefined}
}%
\providecommand \@ifnum [1]{%
 \ifnum #1\expandafter \@firstoftwo
 \else \expandafter \@secondoftwo
 \fi
}%
\providecommand \@ifx [1]{%
 \ifx #1\expandafter \@firstoftwo
 \else \expandafter \@secondoftwo
 \fi
}%
\providecommand \natexlab [1]{#1}%
\providecommand \enquote  [1]{``#1''}%
\providecommand \bibnamefont  [1]{#1}%
\providecommand \bibfnamefont [1]{#1}%
\providecommand \citenamefont [1]{#1}%
\providecommand \href@noop [0]{\@secondoftwo}%
\providecommand \href [0]{\begingroup \@sanitize@url \@href}%
\providecommand \@href[1]{\@@startlink{#1}\@@href}%
\providecommand \@@href[1]{\endgroup#1\@@endlink}%
\providecommand \@sanitize@url [0]{\catcode `\\12\catcode `\$12\catcode
  `\&12\catcode `\#12\catcode `\^12\catcode `\_12\catcode `\%12\relax}%
\providecommand \@@startlink[1]{}%
\providecommand \@@endlink[0]{}%
\providecommand \url  [0]{\begingroup\@sanitize@url \@url }%
\providecommand \@url [1]{\endgroup\@href {#1}{\urlprefix }}%
\providecommand \urlprefix  [0]{URL }%
\providecommand \Eprint [0]{\href }%
\providecommand \doibase [0]{http://dx.doi.org/}%
\providecommand \selectlanguage [0]{\@gobble}%
\providecommand \bibinfo  [0]{\@secondoftwo}%
\providecommand \bibfield  [0]{\@secondoftwo}%
\providecommand \translation [1]{[#1]}%
\providecommand \BibitemOpen [0]{}%
\providecommand \bibitemStop [0]{}%
\providecommand \bibitemNoStop [0]{.\EOS\space}%
\providecommand \EOS [0]{\spacefactor3000\relax}%
\providecommand \BibitemShut  [1]{\csname bibitem#1\endcsname}%
\let\auto@bib@innerbib\@empty
\bibitem [{\citenamefont {Wiesner}(1983)}]{Wiesner83}%
  \BibitemOpen
  \bibfield  {author} {\bibinfo {author} {\bibfnamefont {S.}~\bibnamefont
  {Wiesner}},\ }\href@noop {} {\bibfield  {journal} {\bibinfo  {journal}
  {SIGACT News}\ }\textbf {\bibinfo {volume} {15}},\ \bibinfo {pages} {78}
  (\bibinfo {year} {1983})}\BibitemShut {NoStop}%
\bibitem [{\citenamefont {Bennett}\ and\ \citenamefont
  {Brassard}(1984)}]{Bennett84}%
  \BibitemOpen
  \bibfield  {author} {\bibinfo {author} {\bibfnamefont {C.~H.}\ \bibnamefont
  {Bennett}}\ and\ \bibinfo {author} {\bibfnamefont {G.}~\bibnamefont
  {Brassard}},\ }\href@noop {} {\bibfield  {journal} {\bibinfo  {journal}
  {Proceedings of IEEE International Conference on Computers, Systems and
  Signal Processing}\ ,\ \bibinfo {pages} {175}} (\bibinfo {year}
  {1984})}\BibitemShut {NoStop}%
\bibitem [{\citenamefont {Ekert}(1991)}]{Ekert91}%
  \BibitemOpen
  \bibfield  {author} {\bibinfo {author} {\bibfnamefont {A.}~\bibnamefont
  {Ekert}},\ }\href@noop {} {\bibfield  {journal} {\bibinfo  {journal}
  {Physical Review Letters}\ }\textbf {\bibinfo {volume} {67}},\ \bibinfo
  {pages} {661} (\bibinfo {year} {1991})}\BibitemShut {NoStop}%
\bibitem [{\citenamefont {Mayers}(1997)}]{mayers1997}%
  \BibitemOpen
  \bibfield  {author} {\bibinfo {author} {\bibfnamefont {D.}~\bibnamefont
  {Mayers}},\ }\href@noop {} {\bibfield  {journal} {\bibinfo  {journal}
  {Physical Review Letters}\ }\textbf {\bibinfo {volume} {78}},\ \bibinfo
  {pages} {3414} (\bibinfo {year} {1997})}\BibitemShut {NoStop}%
\bibitem [{\citenamefont {Mayers}(1996)}]{mayers:trouble}%
  \BibitemOpen
  \bibfield  {author} {\bibinfo {author} {\bibfnamefont {D.}~\bibnamefont
  {Mayers}},\ }\href@noop {} {\bibfield  {journal} {\bibinfo  {journal} {arXiv
  preprint, arxiv:quant-ph/9603015}\ } (\bibinfo {year} {1996})}\BibitemShut
  {NoStop}%
\bibitem [{\citenamefont {Lo}\ and\ \citenamefont
  {Chau}(1997)}]{lo&chau:bitcom}%
  \BibitemOpen
  \bibfield  {author} {\bibinfo {author} {\bibfnamefont {H.-K.}\ \bibnamefont
  {Lo}}\ and\ \bibinfo {author} {\bibfnamefont {H.~F.}\ \bibnamefont {Chau}},\
  }\href@noop {} {\bibfield  {journal} {\bibinfo  {journal} {Physical Review
  Letter}\ }\textbf {\bibinfo {volume} {78}},\ \bibinfo {pages} {3410}
  (\bibinfo {year} {1997})}\BibitemShut {NoStop}%
\bibitem [{\citenamefont {Lo}\ and\ \citenamefont
  {Chau}(1998)}]{lo&chau:bitcom2}%
  \BibitemOpen
  \bibfield  {author} {\bibinfo {author} {\bibfnamefont {H.-K.}\ \bibnamefont
  {Lo}}\ and\ \bibinfo {author} {\bibfnamefont {H.~F.}\ \bibnamefont {Chau}},\
  }\href@noop {} {\bibfield  {journal} {\bibinfo  {journal} {Physica D:
  Nonlinear Phenomena}\ }\textbf {\bibinfo {volume} {120}},\ \bibinfo {pages}
  {177} (\bibinfo {year} {1998})}\BibitemShut {NoStop}%
\bibitem [{\citenamefont {Lo}(1997)}]{lo1997}%
  \BibitemOpen
  \bibfield  {author} {\bibinfo {author} {\bibfnamefont {H.-K.}\ \bibnamefont
  {Lo}},\ }\href@noop {} {\bibfield  {journal} {\bibinfo  {journal} {Physical
  Review A}\ }\textbf {\bibinfo {volume} {56}},\ \bibinfo {pages} {1154}
  (\bibinfo {year} {1997})}\BibitemShut {NoStop}%
\bibitem [{\citenamefont {D'Ariano}\ \emph {et~al.}(2007)\citenamefont
  {D'Ariano}, \citenamefont {Kretschmann}, \citenamefont {Schlingemann},\ and\
  \citenamefont {Werner}}]{kretch:bc}%
  \BibitemOpen
  \bibfield  {author} {\bibinfo {author} {\bibfnamefont {G.}~\bibnamefont
  {D'Ariano}}, \bibinfo {author} {\bibfnamefont {D.}~\bibnamefont
  {Kretschmann}}, \bibinfo {author} {\bibfnamefont {D.}~\bibnamefont
  {Schlingemann}}, \ and\ \bibinfo {author} {\bibfnamefont {R.}~\bibnamefont
  {Werner}},\ }\href@noop {} {\bibfield  {journal} {\bibinfo  {journal}
  {Physical Review A}\ }\textbf {\bibinfo {volume} {76}},\ \bibinfo {pages}
  {032328} (\bibinfo {year} {2007})}\BibitemShut {NoStop}%
\bibitem [{\citenamefont {Buhrman}\ \emph {et~al.}(2012)\citenamefont
  {Buhrman}, \citenamefont {Christandl},\ and\ \citenamefont
  {Schaffner}}]{buhrman2012complete}%
  \BibitemOpen
  \bibfield  {author} {\bibinfo {author} {\bibfnamefont {H.}~\bibnamefont
  {Buhrman}}, \bibinfo {author} {\bibfnamefont {M.}~\bibnamefont {Christandl}},
  \ and\ \bibinfo {author} {\bibfnamefont {C.}~\bibnamefont {Schaffner}},\
  }\href@noop {} {\bibfield  {journal} {\bibinfo  {journal} {Physical Review
  Letters}\ }\textbf {\bibinfo {volume} {109}},\ \bibinfo {pages} {160501}
  (\bibinfo {year} {2012})}\BibitemShut {NoStop}%
\bibitem [{\citenamefont {Maurer}(1992)}]{Maurer92b}%
  \BibitemOpen
  \bibfield  {author} {\bibinfo {author} {\bibfnamefont {U.}~\bibnamefont
  {Maurer}},\ }\href@noop {} {\bibfield  {journal} {\bibinfo  {journal}
  {Journal of Cryptology}\ }\textbf {\bibinfo {volume} {5}},\ \bibinfo {pages}
  {53} (\bibinfo {year} {1992})}\BibitemShut {NoStop}%
\bibitem [{\citenamefont {Cachin}\ and\ \citenamefont
  {Maurer}(1997)}]{cachin:bounded}%
  \BibitemOpen
  \bibfield  {author} {\bibinfo {author} {\bibfnamefont {C.}~\bibnamefont
  {Cachin}}\ and\ \bibinfo {author} {\bibfnamefont {U.~M.}\ \bibnamefont
  {Maurer}},\ }in\ \href@noop {} {\emph {\bibinfo {booktitle} {Proceedings of
  CRYPTO 1997}}},\ \bibinfo {series and number} {Lecture Notes in Computer
  Science}\ (\bibinfo {year} {1997})\ pp.\ \bibinfo {pages}
  {292--306}\BibitemShut {NoStop}%
\bibitem [{\citenamefont {Damg{\aa}rd}\ \emph {et~al.}(2008)\citenamefont
  {Damg{\aa}rd}, \citenamefont {Fehr}, \citenamefont {Salvail},\ and\
  \citenamefont {Schaffner}}]{damgaard2008}%
  \BibitemOpen
  \bibfield  {author} {\bibinfo {author} {\bibfnamefont {I.~B.}\ \bibnamefont
  {Damg{\aa}rd}}, \bibinfo {author} {\bibfnamefont {S.}~\bibnamefont {Fehr}},
  \bibinfo {author} {\bibfnamefont {L.}~\bibnamefont {Salvail}}, \ and\
  \bibinfo {author} {\bibfnamefont {C.}~\bibnamefont {Schaffner}},\ }\href@noop
  {} {\bibfield  {journal} {\bibinfo  {journal} {SIAM Journal on Computing}\
  }\textbf {\bibinfo {volume} {37}},\ \bibinfo {pages} {1865} (\bibinfo {year}
  {2008})}\BibitemShut {NoStop}%
\bibitem [{\citenamefont {Damg{\aa}rd}\ \emph {et~al.}(2007)\citenamefont
  {Damg{\aa}rd}, \citenamefont {Fehr}, \citenamefont {Renner}, \citenamefont
  {Salvail},\ and\ \citenamefont {Schaffner}}]{damgaard2007}%
  \BibitemOpen
  \bibfield  {author} {\bibinfo {author} {\bibfnamefont {I.~B.}\ \bibnamefont
  {Damg{\aa}rd}}, \bibinfo {author} {\bibfnamefont {S.}~\bibnamefont {Fehr}},
  \bibinfo {author} {\bibfnamefont {R.}~\bibnamefont {Renner}}, \bibinfo
  {author} {\bibfnamefont {L.}~\bibnamefont {Salvail}}, \ and\ \bibinfo
  {author} {\bibfnamefont {C.}~\bibnamefont {Schaffner}},\ }in\ \href@noop {}
  {\emph {\bibinfo {booktitle} {Advances in Cryptology-CRYPTO 2007}}}\
  (\bibinfo  {publisher} {Springer},\ \bibinfo {year} {2007})\ pp.\ \bibinfo
  {pages} {360--378}\BibitemShut {NoStop}%
\bibitem [{\citenamefont {Wehner}\ \emph {et~al.}(2008)\citenamefont {Wehner},
  \citenamefont {Schaffner},\ and\ \citenamefont {Terhal}}]{wehner2008}%
  \BibitemOpen
  \bibfield  {author} {\bibinfo {author} {\bibfnamefont {S.}~\bibnamefont
  {Wehner}}, \bibinfo {author} {\bibfnamefont {C.}~\bibnamefont {Schaffner}}, \
  and\ \bibinfo {author} {\bibfnamefont {B.~M.}\ \bibnamefont {Terhal}},\
  }\href@noop {} {\bibfield  {journal} {\bibinfo  {journal} {Physical Review
  Letters}\ }\textbf {\bibinfo {volume} {100}},\ \bibinfo {pages} {220502}
  (\bibinfo {year} {2008})}\BibitemShut {NoStop}%
\bibitem [{\citenamefont {Konig}\ \emph {et~al.}(2012)\citenamefont {Konig},
  \citenamefont {Wehner},\ and\ \citenamefont {Wullschleger}}]{Koenig2012}%
  \BibitemOpen
  \bibfield  {author} {\bibinfo {author} {\bibfnamefont {R.}~\bibnamefont
  {Konig}}, \bibinfo {author} {\bibfnamefont {S.}~\bibnamefont {Wehner}}, \
  and\ \bibinfo {author} {\bibfnamefont {J.}~\bibnamefont {Wullschleger}},\
  }\href@noop {} {\bibfield  {journal} {\bibinfo  {journal} {IEEE Transactions
  on Information Theory}\ }\textbf {\bibinfo {volume} {58}},\ \bibinfo {pages}
  {1962} (\bibinfo {year} {2012})}\BibitemShut {NoStop}%
\bibitem [{\citenamefont {Berta}\ \emph {et~al.}(2012)\citenamefont {Berta},
  \citenamefont {Fawzi},\ and\ \citenamefont {Wehner}}]{Berta2012}%
  \BibitemOpen
  \bibfield  {author} {\bibinfo {author} {\bibfnamefont {M.}~\bibnamefont
  {Berta}}, \bibinfo {author} {\bibfnamefont {O.}~\bibnamefont {Fawzi}}, \ and\
  \bibinfo {author} {\bibfnamefont {S.}~\bibnamefont {Wehner}},\ }in\
  \href@noop {} {\emph {\bibinfo {booktitle} {Advances in Cryptology CRYPTO
  2012}}},\ \bibinfo {series} {Lecture Notes in Computer Science}, Vol.\
  \bibinfo {volume} {7417}\ (\bibinfo {year} {2012})\ pp.\ \bibinfo {pages}
  {776--793}\BibitemShut {NoStop}%
\bibitem [{\citenamefont {Berta}\ \emph {et~al.}(2013)\citenamefont {Berta},
  \citenamefont {Brandao}, \citenamefont {Christandl},\ and\ \citenamefont
  {Wehner}}]{berta2013}%
  \BibitemOpen
  \bibfield  {author} {\bibinfo {author} {\bibfnamefont {M.}~\bibnamefont
  {Berta}}, \bibinfo {author} {\bibfnamefont {F.~G.}\ \bibnamefont {Brandao}},
  \bibinfo {author} {\bibfnamefont {M.}~\bibnamefont {Christandl}}, \ and\
  \bibinfo {author} {\bibfnamefont {S.}~\bibnamefont {Wehner}},\ }\href@noop {}
  {\bibfield  {journal} {\bibinfo  {journal} {IEEE Transactions on Information
  Theory}\ }\textbf {\bibinfo {volume} {59}},\ \bibinfo {pages} {6779}
  (\bibinfo {year} {2013})}\BibitemShut {NoStop}%
\bibitem [{\citenamefont {Dupuis}\ \emph {et~al.}(2015)\citenamefont {Dupuis},
  \citenamefont {Fawzi},\ and\ \citenamefont {Wehner}}]{Dupuis2015}%
  \BibitemOpen
  \bibfield  {author} {\bibinfo {author} {\bibfnamefont {F.}~\bibnamefont
  {Dupuis}}, \bibinfo {author} {\bibfnamefont {O.}~\bibnamefont {Fawzi}}, \
  and\ \bibinfo {author} {\bibfnamefont {S.}~\bibnamefont {Wehner}},\
  }\href@noop {} {\bibfield  {journal} {\bibinfo  {journal} {IEEE Transactions
  on Information Theory}\ }\textbf {\bibinfo {volume} {61}},\ \bibinfo {pages}
  {1093} (\bibinfo {year} {2015})}\BibitemShut {NoStop}%
\bibitem [{\citenamefont {Ng}\ \emph {et~al.}(2012{\natexlab{a}})\citenamefont
  {Ng}, \citenamefont {Joshi}, \citenamefont {Ming}, \citenamefont
  {Kurtsiefer},\ and\ \citenamefont {Wehner}}]{ng2012}%
  \BibitemOpen
  \bibfield  {author} {\bibinfo {author} {\bibfnamefont {N.~H.~Y.}\
  \bibnamefont {Ng}}, \bibinfo {author} {\bibfnamefont {S.~K.}\ \bibnamefont
  {Joshi}}, \bibinfo {author} {\bibfnamefont {C.~C.}\ \bibnamefont {Ming}},
  \bibinfo {author} {\bibfnamefont {C.}~\bibnamefont {Kurtsiefer}}, \ and\
  \bibinfo {author} {\bibfnamefont {S.}~\bibnamefont {Wehner}},\ }\href@noop {}
  {\bibfield  {journal} {\bibinfo  {journal} {Nature Communications}\ }\textbf
  {\bibinfo {volume} {3}},\ \bibinfo {pages} {1326} (\bibinfo {year}
  {2012}{\natexlab{a}})}\BibitemShut {NoStop}%
\bibitem [{\citenamefont {Erven}\ \emph {et~al.}(2014)\citenamefont {Erven},
  \citenamefont {Ng}, \citenamefont {Gigov}, \citenamefont {Laflamme},
  \citenamefont {Wehner},\ and\ \citenamefont {Weihs}}]{erven2014experimental}%
  \BibitemOpen
  \bibfield  {author} {\bibinfo {author} {\bibfnamefont {C.}~\bibnamefont
  {Erven}}, \bibinfo {author} {\bibfnamefont {N.~H.~Y.}\ \bibnamefont {Ng}},
  \bibinfo {author} {\bibfnamefont {N.}~\bibnamefont {Gigov}}, \bibinfo
  {author} {\bibfnamefont {R.}~\bibnamefont {Laflamme}}, \bibinfo {author}
  {\bibfnamefont {S.}~\bibnamefont {Wehner}}, \ and\ \bibinfo {author}
  {\bibfnamefont {G.}~\bibnamefont {Weihs}},\ }\href@noop {} {\bibfield
  {journal} {\bibinfo  {journal} {Nature Communications}\ }\textbf {\bibinfo
  {volume} {5}} (\bibinfo {year} {2014})}\BibitemShut {NoStop}%
\bibitem [{\citenamefont {Wehner}\ \emph {et~al.}(2010)\citenamefont {Wehner},
  \citenamefont {Curty}, \citenamefont {Schaffner},\ and\ \citenamefont
  {Lo}}]{wcsl2010}%
  \BibitemOpen
  \bibfield  {author} {\bibinfo {author} {\bibfnamefont {S.}~\bibnamefont
  {Wehner}}, \bibinfo {author} {\bibfnamefont {M.}~\bibnamefont {Curty}},
  \bibinfo {author} {\bibfnamefont {C.}~\bibnamefont {Schaffner}}, \ and\
  \bibinfo {author} {\bibfnamefont {H.-K.}\ \bibnamefont {Lo}},\ }\href@noop {}
  {\bibfield  {journal} {\bibinfo  {journal} {Physical Review A}\ }\textbf
  {\bibinfo {volume} {81}},\ \bibinfo {pages} {052336} (\bibinfo {year}
  {2010})},\ \bibinfo {note} {arXiv:0911.2302v2}\BibitemShut {NoStop}%
\bibitem [{\citenamefont {Schaffner}(2010)}]{schaffner2010}%
  \BibitemOpen
  \bibfield  {author} {\bibinfo {author} {\bibfnamefont {C.}~\bibnamefont
  {Schaffner}},\ }\href@noop {} {\bibfield  {journal} {\bibinfo  {journal}
  {Physical Review A}\ }\textbf {\bibinfo {volume} {82}},\ \bibinfo {pages}
  {032308} (\bibinfo {year} {2010})}\BibitemShut {NoStop}%
\bibitem [{\citenamefont {Lo}\ \emph {et~al.}(2014)\citenamefont {Lo},
  \citenamefont {Curty},\ and\ \citenamefont {Tamaki}}]{lo2014}%
  \BibitemOpen
  \bibfield  {author} {\bibinfo {author} {\bibfnamefont {H.-K.}\ \bibnamefont
  {Lo}}, \bibinfo {author} {\bibfnamefont {M.}~\bibnamefont {Curty}}, \ and\
  \bibinfo {author} {\bibfnamefont {K.}~\bibnamefont {Tamaki}},\ }\href@noop {}
  {\bibfield  {journal} {\bibinfo  {journal} {Nature Photonics}\ }\textbf
  {\bibinfo {volume} {8}},\ \bibinfo {pages} {595} (\bibinfo {year}
  {2014})}\BibitemShut {NoStop}%
\bibitem [{\citenamefont {Weedbrook}\ \emph {et~al.}(2012)\citenamefont
  {Weedbrook}, \citenamefont {Pirandola}, \citenamefont
  {Garc\'{i}a-Patr\'{o}n}, \citenamefont {Cerf}, \citenamefont {Ralph},
  \citenamefont {Shapiro},\ and\ \citenamefont {Lloyd}}]{Weedbrook12}%
  \BibitemOpen
  \bibfield  {author} {\bibinfo {author} {\bibfnamefont {C.}~\bibnamefont
  {Weedbrook}}, \bibinfo {author} {\bibfnamefont {S.}~\bibnamefont
  {Pirandola}}, \bibinfo {author} {\bibfnamefont {R.}~\bibnamefont
  {Garc\'{i}a-Patr\'{o}n}}, \bibinfo {author} {\bibfnamefont {N.~J.}\
  \bibnamefont {Cerf}}, \bibinfo {author} {\bibfnamefont {T.~C.}\ \bibnamefont
  {Ralph}}, \bibinfo {author} {\bibfnamefont {J.~H.}\ \bibnamefont {Shapiro}},
  \ and\ \bibinfo {author} {\bibfnamefont {S.}~\bibnamefont {Lloyd}},\
  }\href@noop {} {\bibfield  {journal} {\bibinfo  {journal} {Reviews of Modern
  Physics}\ }\textbf {\bibinfo {volume} {84}},\ \bibinfo {pages} {621}
  (\bibinfo {year} {2012})}\BibitemShut {NoStop}%
\bibitem [{\citenamefont {Masada}\ \emph {et~al.}(2015)\citenamefont {Masada},
  \citenamefont {Miyata}, \citenamefont {Politi}, \citenamefont {Hashimoto},
  \citenamefont {O'Brien},\ and\ \citenamefont {Furusawa}}]{masada2015}%
  \BibitemOpen
  \bibfield  {author} {\bibinfo {author} {\bibfnamefont {G.}~\bibnamefont
  {Masada}}, \bibinfo {author} {\bibfnamefont {K.}~\bibnamefont {Miyata}},
  \bibinfo {author} {\bibfnamefont {A.}~\bibnamefont {Politi}}, \bibinfo
  {author} {\bibfnamefont {T.}~\bibnamefont {Hashimoto}}, \bibinfo {author}
  {\bibfnamefont {J.~L.}\ \bibnamefont {O'Brien}}, \ and\ \bibinfo {author}
  {\bibfnamefont {A.}~\bibnamefont {Furusawa}},\ }\href@noop {} {\bibfield
  {journal} {\bibinfo  {journal} {Nature Photonics}\ }\textbf {\bibinfo
  {volume} {9}},\ \bibinfo {pages} {316} (\bibinfo {year} {2015})}\BibitemShut
  {NoStop}%
\bibitem [{\citenamefont {Rudnicki}(2015)}]{rudnicki2015}%
  \BibitemOpen
  \bibfield  {author} {\bibinfo {author} {\bibfnamefont {{\L}.}~\bibnamefont
  {Rudnicki}},\ }\href@noop {} {\bibfield  {journal} {\bibinfo  {journal}
  {Physical Review A}\ }\textbf {\bibinfo {volume} {91}},\ \bibinfo {pages}
  {032123} (\bibinfo {year} {2015})}\BibitemShut {NoStop}%
\bibitem [{\citenamefont {Kilian}(1988)}]{kilian1988founding}%
  \BibitemOpen
  \bibfield  {author} {\bibinfo {author} {\bibfnamefont {J.}~\bibnamefont
  {Kilian}},\ }in\ \href@noop {} {\emph {\bibinfo {booktitle} {Proceedings of
  the Twentieth Annual ACM Symposium on Theory of Computing}}}\ (\bibinfo
  {organization} {ACM},\ \bibinfo {year} {1988})\ pp.\ \bibinfo {pages}
  {20--31}\BibitemShut {NoStop}%
\bibitem [{\citenamefont {Einstein}\ \emph {et~al.}(1935)\citenamefont
  {Einstein}, \citenamefont {Podolsky},\ and\ \citenamefont {Rosen}}]{epr35}%
  \BibitemOpen
  \bibfield  {author} {\bibinfo {author} {\bibfnamefont {A.}~\bibnamefont
  {Einstein}}, \bibinfo {author} {\bibfnamefont {B.}~\bibnamefont {Podolsky}},
  \ and\ \bibinfo {author} {\bibfnamefont {N.}~\bibnamefont {Rosen}},\
  }\href@noop {} {\bibfield  {journal} {\bibinfo  {journal} {Physical Review
  Letters}\ }\textbf {\bibinfo {volume} {47}},\ \bibinfo {pages} {777}
  (\bibinfo {year} {1935})}\BibitemShut {NoStop}%
\bibitem [{\citenamefont {Furrer}\ \emph {et~al.}(2014)\citenamefont {Furrer},
  \citenamefont {Berta}, \citenamefont {Tomamichel}, \citenamefont {Scholz},\
  and\ \citenamefont {Christandl}}]{furrer2014PQUR}%
  \BibitemOpen
  \bibfield  {author} {\bibinfo {author} {\bibfnamefont {F.}~\bibnamefont
  {Furrer}}, \bibinfo {author} {\bibfnamefont {M.}~\bibnamefont {Berta}},
  \bibinfo {author} {\bibfnamefont {M.}~\bibnamefont {Tomamichel}}, \bibinfo
  {author} {\bibfnamefont {V.~B.}\ \bibnamefont {Scholz}}, \ and\ \bibinfo
  {author} {\bibfnamefont {M.}~\bibnamefont {Christandl}},\ }\href@noop {}
  {\bibfield  {journal} {\bibinfo  {journal} {Journal of Mathematical Physics}\
  }\textbf {\bibinfo {volume} {55}},\ \bibinfo {pages} {122205} (\bibinfo
  {year} {2014})}\BibitemShut {NoStop}%
\bibitem [{Note1()}]{Note1}%
  \BibitemOpen
  \bibinfo {note} {Note that due to transmission losses in Bob's mode, he
  scales his outcomes with $1/\protect \sqrt {\tau }$, where $\tau $ is the
  transmissivity.}\BibitemShut {Stop}%
\bibitem [{Note2()}]{Note2}%
  \BibitemOpen
  \bibinfo {note} {Note that we could introduce an additional step that Bob can
  check if the error correction worked properly, namely, by Alice sending a
  hash of $Z_0,Z_1$. However, Bob is not allowed to tell Alice whether the test
  was passed or not, since Alice could design attacks which lead to pass or
  failure of the test depending on his choice $t$.}\BibitemShut {Stop}%
\bibitem [{\citenamefont {Renner}(2005)}]{renner05}%
  \BibitemOpen
  \bibfield  {author} {\bibinfo {author} {\bibfnamefont {R.}~\bibnamefont
  {Renner}},\ }\emph {\bibinfo {title} {{Security of Quantum Key
  Distribution}}},\ \href@noop {} {Ph.D. thesis},\ \bibinfo  {school} {ETH
  Zurich} (\bibinfo {year} {2005})\BibitemShut {NoStop}%
\bibitem [{\citenamefont {K\"onig}\ \emph {et~al.}(2009)\citenamefont
  {K\"onig}, \citenamefont {Renner},\ and\ \citenamefont
  {Schaffner}}]{koenig-2008}%
  \BibitemOpen
  \bibfield  {author} {\bibinfo {author} {\bibfnamefont {R.}~\bibnamefont
  {K\"onig}}, \bibinfo {author} {\bibfnamefont {R.}~\bibnamefont {Renner}}, \
  and\ \bibinfo {author} {\bibfnamefont {C.}~\bibnamefont {Schaffner}},\
  }\href@noop {} {\bibfield  {journal} {\bibinfo  {journal} {IEEE Transactions
  on Information Theory}\ }\textbf {\bibinfo {volume} {55}},\ \bibinfo {pages}
  {4674} (\bibinfo {year} {2009})}\BibitemShut {NoStop}%
\bibitem [{\citenamefont {Wullschleger}(2007)}]{wulli}%
  \BibitemOpen
  \bibfield  {author} {\bibinfo {author} {\bibfnamefont {J.}~\bibnamefont
  {Wullschleger}},\ }in\ \href@noop {} {\emph {\bibinfo {booktitle} {Advances
  in Cryptology EUROCRYPT}}},\ \bibinfo {series and number} {Lecture Notes in
  Computer Science}\ (\bibinfo  {publisher} {Springer},\ \bibinfo {year}
  {2007})\ pp.\ \bibinfo {pages} {555--572}\BibitemShut {NoStop}%
\bibitem [{\citenamefont {Vidick}\ and\ \citenamefont
  {Wehner}(2011)}]{splitting}%
  \BibitemOpen
  \bibfield  {author} {\bibinfo {author} {\bibfnamefont {T.}~\bibnamefont
  {Vidick}}\ and\ \bibinfo {author} {\bibfnamefont {S.}~\bibnamefont
  {Wehner}},\ }\href@noop {} {\bibfield  {journal} {\bibinfo  {journal}
  {Physical Review Letters}\ }\textbf {\bibinfo {volume} {030401}} (\bibinfo
  {year} {2011})}\BibitemShut {NoStop}%
\bibitem [{Note3()}]{Note3}%
  \BibitemOpen
  \bibinfo {note} {In information theory this is referred to as a strong
  converse for the classical capacity and has been shown for many
  channels.}\BibitemShut {Stop}%
\bibitem [{\citenamefont {Jouguet}\ \emph {et~al.}(2011)\citenamefont
  {Jouguet}, \citenamefont {Kunz-Jacques},\ and\ \citenamefont
  {Leverrier}}]{jouguet2011}%
  \BibitemOpen
  \bibfield  {author} {\bibinfo {author} {\bibfnamefont {P.}~\bibnamefont
  {Jouguet}}, \bibinfo {author} {\bibfnamefont {S.}~\bibnamefont
  {Kunz-Jacques}}, \ and\ \bibinfo {author} {\bibfnamefont {A.}~\bibnamefont
  {Leverrier}},\ }\href@noop {} {\bibfield  {journal} {\bibinfo  {journal}
  {Physical Review A}\ }\textbf {\bibinfo {volume} {84}},\ \bibinfo {pages}
  {062317} (\bibinfo {year} {2011})}\BibitemShut {NoStop}%
\bibitem [{\citenamefont {Jouguet}\ \emph {et~al.}(2014)\citenamefont
  {Jouguet}, \citenamefont {Elkouss},\ and\ \citenamefont
  {Kunz-Jacques}}]{jouguet2014}%
  \BibitemOpen
  \bibfield  {author} {\bibinfo {author} {\bibfnamefont {P.}~\bibnamefont
  {Jouguet}}, \bibinfo {author} {\bibfnamefont {D.}~\bibnamefont {Elkouss}}, \
  and\ \bibinfo {author} {\bibfnamefont {S.}~\bibnamefont {Kunz-Jacques}},\
  }\href@noop {} {\bibfield  {journal} {\bibinfo  {journal} {Physical Review
  A}\ }\textbf {\bibinfo {volume} {90}},\ \bibinfo {pages} {042329} (\bibinfo
  {year} {2014})}\BibitemShut {NoStop}%
\bibitem [{\citenamefont {Gehring}\ \emph {et~al.}(2014)\citenamefont
  {Gehring}, \citenamefont {H{\"a}ndchen}, \citenamefont {Duhme}, \citenamefont
  {Furrer}, \citenamefont {Franz}, \citenamefont {Pacher}, \citenamefont
  {Werner},\ and\ \citenamefont {Schnabel}}]{gehring2014}%
  \BibitemOpen
  \bibfield  {author} {\bibinfo {author} {\bibfnamefont {T.}~\bibnamefont
  {Gehring}}, \bibinfo {author} {\bibfnamefont {V.}~\bibnamefont
  {H{\"a}ndchen}}, \bibinfo {author} {\bibfnamefont {J.}~\bibnamefont {Duhme}},
  \bibinfo {author} {\bibfnamefont {F.}~\bibnamefont {Furrer}}, \bibinfo
  {author} {\bibfnamefont {T.}~\bibnamefont {Franz}}, \bibinfo {author}
  {\bibfnamefont {C.}~\bibnamefont {Pacher}}, \bibinfo {author} {\bibfnamefont
  {R.~F.}\ \bibnamefont {Werner}}, \ and\ \bibinfo {author} {\bibfnamefont
  {R.}~\bibnamefont {Schnabel}},\ }\href@noop {} {\bibfield  {journal}
  {\bibinfo  {journal} {arxiv preprint, arXiv:1406.6174}\ } (\bibinfo {year}
  {2014})}\BibitemShut {NoStop}%
\bibitem [{\citenamefont {Ng}\ \emph {et~al.}(2012{\natexlab{b}})\citenamefont
  {Ng}, \citenamefont {Berta},\ and\ \citenamefont {Wehner}}]{Nelly12}%
  \BibitemOpen
  \bibfield  {author} {\bibinfo {author} {\bibfnamefont {N.~H.~Y.}\
  \bibnamefont {Ng}}, \bibinfo {author} {\bibfnamefont {M.}~\bibnamefont
  {Berta}}, \ and\ \bibinfo {author} {\bibfnamefont {S.}~\bibnamefont
  {Wehner}},\ }\href@noop {} {\bibfield  {journal} {\bibinfo  {journal} {Phys.
  Rev. A}\ }\textbf {\bibinfo {volume} {86}},\ \bibinfo {pages} {042315}
  (\bibinfo {year} {2012}{\natexlab{b}})}\BibitemShut {NoStop}%
\bibitem [{\citenamefont {Kennard}(1927)}]{kennard1927}%
  \BibitemOpen
  \bibfield  {author} {\bibinfo {author} {\bibfnamefont {E.}~\bibnamefont
  {Kennard}},\ }\href@noop {} {\bibfield  {journal} {\bibinfo  {journal}
  {Zeitschrift f{\"u}r Physik}\ }\textbf {\bibinfo {volume} {44}},\ \bibinfo
  {pages} {326} (\bibinfo {year} {1927})}\BibitemShut {NoStop}%
\bibitem [{\citenamefont {Bialynicki-Birula}(1984)}]{Birula84}%
  \BibitemOpen
  \bibfield  {author} {\bibinfo {author} {\bibfnamefont {I.}~\bibnamefont
  {Bialynicki-Birula}},\ }\href@noop {} {\bibfield  {journal} {\bibinfo
  {journal} {Physics Letters}\ }\textbf {\bibinfo {volume} {103}},\ \bibinfo
  {pages} {253} (\bibinfo {year} {1984})}\BibitemShut {NoStop}%
\bibitem [{\citenamefont {Landau}\ and\ \citenamefont
  {Pollak}(1961)}]{Landau61}%
  \BibitemOpen
  \bibfield  {author} {\bibinfo {author} {\bibfnamefont {H.~J.}\ \bibnamefont
  {Landau}}\ and\ \bibinfo {author} {\bibfnamefont {H.~O.}\ \bibnamefont
  {Pollak}},\ }\href@noop {} {\bibfield  {journal} {\bibinfo  {journal} {The
  Bell System Technical Journal}\ }\textbf {\bibinfo {volume} {65}},\ \bibinfo
  {pages} {43} (\bibinfo {year} {1961})}\BibitemShut {NoStop}%
\bibitem [{\citenamefont {Furrer}\ \emph {et~al.}(2011)\citenamefont {Furrer},
  \citenamefont {Aberg},\ and\ \citenamefont {Renner}}]{Furrer10}%
  \BibitemOpen
  \bibfield  {author} {\bibinfo {author} {\bibfnamefont {F.}~\bibnamefont
  {Furrer}}, \bibinfo {author} {\bibfnamefont {J.}~\bibnamefont {Aberg}}, \
  and\ \bibinfo {author} {\bibfnamefont {R.}~\bibnamefont {Renner}},\
  }\href@noop {} {\bibfield  {journal} {\bibinfo  {journal} {Communications in
  Mathematical Physics}\ }\textbf {\bibinfo {volume} {306}},\ \bibinfo {pages}
  {165} (\bibinfo {year} {2011})}\BibitemShut {NoStop}%
\bibitem [{\citenamefont {Giovannetti}\ \emph
  {et~al.}(2013{\natexlab{a}})\citenamefont {Giovannetti}, \citenamefont
  {Holevo},\ and\ \citenamefont {Garcia-Patron}}]{giovannetti2013A}%
  \BibitemOpen
  \bibfield  {author} {\bibinfo {author} {\bibfnamefont {V.}~\bibnamefont
  {Giovannetti}}, \bibinfo {author} {\bibfnamefont {A.}~\bibnamefont {Holevo}},
  \ and\ \bibinfo {author} {\bibfnamefont {R.}~\bibnamefont {Garcia-Patron}},\
  }\href@noop {} {\bibfield  {journal} {\bibinfo  {journal} {arXiv preprint,
  arXiv:1312.2251}\ } (\bibinfo {year} {2013}{\natexlab{a}})}\BibitemShut
  {NoStop}%
\bibitem [{\citenamefont {Giovannetti}\ \emph
  {et~al.}(2013{\natexlab{b}})\citenamefont {Giovannetti}, \citenamefont
  {Garcia-Patron}, \citenamefont {Cerf},\ and\ \citenamefont
  {Holevo}}]{giovannettiB}%
  \BibitemOpen
  \bibfield  {author} {\bibinfo {author} {\bibfnamefont {V.}~\bibnamefont
  {Giovannetti}}, \bibinfo {author} {\bibfnamefont {R.}~\bibnamefont
  {Garcia-Patron}}, \bibinfo {author} {\bibfnamefont {N.}~\bibnamefont {Cerf}},
  \ and\ \bibinfo {author} {\bibfnamefont {A.}~\bibnamefont {Holevo}},\
  }\href@noop {} {\bibfield  {journal} {\bibinfo  {journal} {arXiv preprint,
  arXiv:1312.6225}\ } (\bibinfo {year} {2013}{\natexlab{b}})}\BibitemShut
  {NoStop}%
\bibitem [{\citenamefont {Wilde}\ and\ \citenamefont
  {Winter}(2014)}]{wilde2014}%
  \BibitemOpen
  \bibfield  {author} {\bibinfo {author} {\bibfnamefont {M.~M.}\ \bibnamefont
  {Wilde}}\ and\ \bibinfo {author} {\bibfnamefont {A.}~\bibnamefont {Winter}},\
  }\href@noop {} {\bibfield  {journal} {\bibinfo  {journal} {Problems of
  Information Transmission}\ }\textbf {\bibinfo {volume} {50}},\ \bibinfo
  {pages} {117} (\bibinfo {year} {2014})}\BibitemShut {NoStop}%
\bibitem [{\citenamefont {Tomamichel}\ \emph
  {et~al.}(2010{\natexlab{a}})\citenamefont {Tomamichel}, \citenamefont
  {Colbeck},\ and\ \citenamefont {Renner}}]{Tomamichel09}%
  \BibitemOpen
  \bibfield  {author} {\bibinfo {author} {\bibfnamefont {M.}~\bibnamefont
  {Tomamichel}}, \bibinfo {author} {\bibfnamefont {R.}~\bibnamefont {Colbeck}},
  \ and\ \bibinfo {author} {\bibfnamefont {R.}~\bibnamefont {Renner}},\
  }\href@noop {} {\bibfield  {journal} {\bibinfo  {journal} {IEEE Transactions
  on Information Theory}\ }\textbf {\bibinfo {volume} {56}},\ \bibinfo {pages}
  {4674} (\bibinfo {year} {2010}{\natexlab{a}})}\BibitemShut {NoStop}%
\bibitem [{\citenamefont {Tomamichel}(2013)}]{tomamichel:thesis}%
  \BibitemOpen
  \bibfield  {author} {\bibinfo {author} {\bibfnamefont {M.}~\bibnamefont
  {Tomamichel}},\ }\emph {\bibinfo {title} {A Framework for Non-Asymptotic
  Quantum Information Theory}},\ \href@noop {} {Ph.D. thesis},\ \bibinfo
  {school} {ETH Z\"urich} (\bibinfo {year} {2013})\BibitemShut {NoStop}%
\bibitem [{\citenamefont {Berta}\ \emph {et~al.}(2011)\citenamefont {Berta},
  \citenamefont {Furrer},\ and\ \citenamefont {Scholz}}]{berta2011}%
  \BibitemOpen
  \bibfield  {author} {\bibinfo {author} {\bibfnamefont {M.}~\bibnamefont
  {Berta}}, \bibinfo {author} {\bibfnamefont {F.}~\bibnamefont {Furrer}}, \
  and\ \bibinfo {author} {\bibfnamefont {V.~B.}\ \bibnamefont {Scholz}},\
  }\href@noop {} {\bibfield  {journal} {\bibinfo  {journal} {arXiv preprint,
  arXiv:1107.5460}\ } (\bibinfo {year} {2011})}\BibitemShut {NoStop}%
\bibitem [{\citenamefont {Tomamichel}\ \emph {et~al.}(2009)\citenamefont
  {Tomamichel}, \citenamefont {Colbeck},\ and\ \citenamefont
  {Renner}}]{Tomamichel08}%
  \BibitemOpen
  \bibfield  {author} {\bibinfo {author} {\bibfnamefont {M.}~\bibnamefont
  {Tomamichel}}, \bibinfo {author} {\bibfnamefont {R.}~\bibnamefont {Colbeck}},
  \ and\ \bibinfo {author} {\bibfnamefont {R.}~\bibnamefont {Renner}},\
  }\href@noop {} {\bibfield  {journal} {\bibinfo  {journal} {IEEE Transactions
  on Information Theory}\ }\textbf {\bibinfo {volume} {55}},\ \bibinfo {pages}
  {5840} (\bibinfo {year} {2009})}\BibitemShut {NoStop}%
\bibitem [{\citenamefont {Dym}\ and\ \citenamefont {McKean}(1972)}]{DymMcKean}%
  \BibitemOpen
  \bibfield  {author} {\bibinfo {author} {\bibfnamefont {H.}~\bibnamefont
  {Dym}}\ and\ \bibinfo {author} {\bibfnamefont {H.~P.}\ \bibnamefont
  {McKean}},\ }\href@noop {} {\emph {\bibinfo {title} {Fourier Series and
  Integrals}}}\ (\bibinfo  {publisher} {Academic, New York},\ \bibinfo {year}
  {1972})\BibitemShut {NoStop}%
\bibitem [{\citenamefont {Bialynicki-Birula}\ and\ \citenamefont
  {Mycielski}(1975)}]{Birula75}%
  \BibitemOpen
  \bibfield  {author} {\bibinfo {author} {\bibfnamefont {I.}~\bibnamefont
  {Bialynicki-Birula}}\ and\ \bibinfo {author} {\bibfnamefont {J.}~\bibnamefont
  {Mycielski}},\ }\href@noop {} {\bibfield  {journal} {\bibinfo  {journal}
  {Communications in Mathematical Physics}\ }\textbf {\bibinfo {volume} {44}},\
  \bibinfo {pages} {129} (\bibinfo {year} {1975})}\BibitemShut {NoStop}%
\bibitem [{\citenamefont {Beckner}(1975)}]{Beckner75}%
  \BibitemOpen
  \bibfield  {author} {\bibinfo {author} {\bibfnamefont {W.}~\bibnamefont
  {Beckner}},\ }\href@noop {} {\bibfield  {journal} {\bibinfo  {journal}
  {Annals of Mathematics}\ }\textbf {\bibinfo {volume} {102}},\ \bibinfo
  {pages} {159} (\bibinfo {year} {1975})}\BibitemShut {NoStop}%
\bibitem [{\citenamefont {Tomamichel}\ \emph
  {et~al.}(2010{\natexlab{b}})\citenamefont {Tomamichel}, \citenamefont
  {Schaffner}, \citenamefont {Smith},\ and\ \citenamefont
  {Renner}}]{Tomamichel10}%
  \BibitemOpen
  \bibfield  {author} {\bibinfo {author} {\bibfnamefont {M.}~\bibnamefont
  {Tomamichel}}, \bibinfo {author} {\bibfnamefont {C.}~\bibnamefont
  {Schaffner}}, \bibinfo {author} {\bibfnamefont {A.}~\bibnamefont {Smith}}, \
  and\ \bibinfo {author} {\bibfnamefont {R.}~\bibnamefont {Renner}},\
  }\href@noop {} {\bibfield  {journal} {\bibinfo  {journal} {Proceedings of
  IEEE Symposium on Information Theory}\ ,\ \bibinfo {pages} {2703}} (\bibinfo
  {year} {2010}{\natexlab{b}})}\BibitemShut {NoStop}%
\bibitem [{\citenamefont {Renner}\ and\ \citenamefont
  {Cirac}(2009)}]{circac09}%
  \BibitemOpen
  \bibfield  {author} {\bibinfo {author} {\bibfnamefont {R.}~\bibnamefont
  {Renner}}\ and\ \bibinfo {author} {\bibfnamefont {J.~I.}\ \bibnamefont
  {Cirac}},\ }\href@noop {} {\bibfield  {journal} {\bibinfo  {journal}
  {Physical Review Letters}\ }\textbf {\bibinfo {volume} {102}},\ \bibinfo
  {pages} {110504} (\bibinfo {year} {2009})}\BibitemShut {NoStop}%
\bibitem [{\citenamefont {Leverrier}\ \emph {et~al.}(2013)\citenamefont
  {Leverrier}, \citenamefont {Garc{\'\i}a-Patr{\'o}n}, \citenamefont {Renner},\
  and\ \citenamefont {Cerf}}]{2013DeF}%
  \BibitemOpen
  \bibfield  {author} {\bibinfo {author} {\bibfnamefont {A.}~\bibnamefont
  {Leverrier}}, \bibinfo {author} {\bibfnamefont {R.}~\bibnamefont
  {Garc{\'\i}a-Patr{\'o}n}}, \bibinfo {author} {\bibfnamefont {R.}~\bibnamefont
  {Renner}}, \ and\ \bibinfo {author} {\bibfnamefont {N.~J.}\ \bibnamefont
  {Cerf}},\ }\href@noop {} {\bibfield  {journal} {\bibinfo  {journal} {Physical
  Review Letters}\ }\textbf {\bibinfo {volume} {110}},\ \bibinfo {pages}
  {030502} (\bibinfo {year} {2013})}\BibitemShut {NoStop}%
\bibitem [{\citenamefont {Slepian}\ and\ \citenamefont
  {Wolf}(1971)}]{Slepian71}%
  \BibitemOpen
  \bibfield  {author} {\bibinfo {author} {\bibfnamefont {D.}~\bibnamefont
  {Slepian}}\ and\ \bibinfo {author} {\bibfnamefont {J.}~\bibnamefont {Wolf}},\
  }\href@noop {} {\bibfield  {journal} {\bibinfo  {journal} {IEEE Transactions
  on Information Theory}\ }\textbf {\bibinfo {volume} {19}},\ \bibinfo {pages}
  {461} (\bibinfo {year} {1971})}\BibitemShut {NoStop}%
\bibitem [{Note4()}]{Note4}%
  \BibitemOpen
  \bibinfo {note} {If the covariance matrix of two Gaussian random variables
  $X$ and $Y$ is denoted by $\Gamma _{XY}$, then $V_{X|Y}= \protect \qopname
  \relax m{det}\Gamma _{XY}/V_Y$.}\BibitemShut {Stop}%
\bibitem [{\citenamefont {Bardhan}\ and\ \citenamefont
  {Wilde}(2014)}]{bardhan2014}%
  \BibitemOpen
  \bibfield  {author} {\bibinfo {author} {\bibfnamefont {B.~R.}\ \bibnamefont
  {Bardhan}}\ and\ \bibinfo {author} {\bibfnamefont {M.~M.}\ \bibnamefont
  {Wilde}},\ }\href@noop {} {\bibfield  {journal} {\bibinfo  {journal}
  {Physical Review A}\ }\textbf {\bibinfo {volume} {89}},\ \bibinfo {pages}
  {022302} (\bibinfo {year} {2014})}\BibitemShut {NoStop}%
\bibitem [{\citenamefont {Bardhan}\ \emph {et~al.}(2014)\citenamefont
  {Bardhan}, \citenamefont {Garcia-Patron}, \citenamefont {Wilde},\ and\
  \citenamefont {Winter}}]{bardhan2014B}%
  \BibitemOpen
  \bibfield  {author} {\bibinfo {author} {\bibfnamefont {B.~R.}\ \bibnamefont
  {Bardhan}}, \bibinfo {author} {\bibfnamefont {R.}~\bibnamefont
  {Garcia-Patron}}, \bibinfo {author} {\bibfnamefont {M.~M.}\ \bibnamefont
  {Wilde}}, \ and\ \bibinfo {author} {\bibfnamefont {A.}~\bibnamefont
  {Winter}},\ }\href@noop {} {\bibfield  {journal} {\bibinfo  {journal} {arXiv
  preprint, arXiv:1401.4161}\ } (\bibinfo {year} {2014})}\BibitemShut {NoStop}%
\end{thebibliography}%
\bibliographystyle{apsrev4-1}

\end{document}